\documentclass[preprint]{imsart}

\RequirePackage{amsthm,amsmath,amsfonts,amssymb}
\RequirePackage[authoryear]{natbib}
\RequirePackage[colorlinks,citecolor=blue,urlcolor=blue]{hyperref}

\RequirePackage{mathrsfs, mathtools}
\RequirePackage{algorithm, algpseudocode, xpatch}

\startlocaldefs
\theoremstyle{plain}
\newtheorem{thm}{Theorem}[section]
\newtheorem{cor}[thm]{Corollary}
\newtheorem{lem}[thm]{Lemma}
\newtheorem{prop}[thm]{Proposition}
\theoremstyle{definition}

\newtheorem{rem}[thm]{Remark}

\DeclarePairedDelimiter{\paren}{\lparen}{\rparen}
\DeclarePairedDelimiter{\abs}{\lvert}{\rvert}
\DeclarePairedDelimiter{\bracket}{[}{]}
\DeclarePairedDelimiter{\curly}{\{}{\}}
\DeclarePairedDelimiter{\norm}{\lVert}{\rVert}
\DeclarePairedDelimiter{\inner}{\langle}{\rangle}

\newcommand{\mbb}[1]{\mathbb{#1}}
\newcommand{\mcal}[1]{\mathcal{#1}}
\newcommand{\mscr}[1]{\mathscr{#1}}

\newcommand{\cross}{\times}
\newcommand{\Id}{\mathrm{Id}}
\newcommand{\intersect}{\cap}
\newcommand{\union}{\cup}

\makeatletter
\xpatchcmd\ALG@step{\arabic{ALG@line}}{\fmtlinenumber{ALG@line}}{}{}
\makeatother 
\let\fmtlinenumber\arabic 
\newcommand{\mathalph}[1]{$\Alph{#1}$}

\endlocaldefs

\begin{document}

\begin{frontmatter}
\title{Extensions of the solidarity principle of the spectral gap for Gibbs samplers to their blocked and collapsed variants}
\runtitle{Extensions of the solidarity principle}

\begin{aug}
\author[A]{\fnms{Xavier}~\snm{Mak}\ead[label=e1]{xaviermak@ufl.edu}}
\and
\author[A]{\fnms{James P.}~\snm{Hobert}\ead[label=e2]{jhobert@stat.ufl.edu}}
\address[A]{Department of Statistics,
	University of Florida\printead[presep={,\ }]{e1,e2}}
\runauthor{X. Mak and J. P. Hobert}
\end{aug}

\begin{abstract}
Connections of a spectral nature are formed between Gibbs samplers and their blocked and collapsed variants. The solidarity principle of the spectral gap for full Gibbs samplers is generalized to different cycles and mixtures of Gibbs steps. This generalized solidarity principle is employed to establish that every cycle and mixture of Gibbs steps, which includes blocked Gibbs samplers and collapsed Gibbs samplers, inherits a spectral gap from a full Gibbs sampler. Exact relations between the spectra corresponding to blocked and collapsed variants of a Gibbs sampler are also established. An example is given to show that a blocked or collapsed Gibbs sampler does not in general inherit geometric ergodicity or a spectral gap from another blocked or collapsed Gibbs sampler. 
\end{abstract}



\end{frontmatter}


\section{Introduction}

Gibbs samplers (\cite{Geman:Geman:1984}, \cite{Gelfand:Smith:1990}) are commonly used Markov chain Monte Carlo (MCMC) algorithms to sample from intractable multi-dimensional probability distributions. The fundamental strategy that Gibbs samplers employ is to update a given vector by sampling one variable or a block of variables at a time through draws from lower-dimensional conditional distributions associated with the target distribution. The nature of these component-wise sampling algorithms can be seen by viewing the Markov kernels of Gibbs samplers as combinations of simpler Markov kernels that each correspond to a single Gibbs step or update. In this perspective, Gibbs samplers are members of a larger family of component-wise algorithms: one comprising of any algorithm with enough of these Gibbs steps assembled so that each coordinate is updated. By varying the choice of Gibbs steps and how these Gibbs steps are combined, there are a number of algorithms to use; some may be more computationally feasible, while others may be more theoretically tractable. 

To better describe the different algorithms, let us informally define the setting in which we view them. (The formal account is provided in Section~\ref{section: Preliminaries}.) Let \(\mcal{X} = \mcal{X}_1 \cross \cdots \cross \mcal{X}_K\), where \(K \geq 2\), denote the state space on which a probability measure \(\pi\) of interest is defined.  On \(L^2(\pi)\), the space of square-integrable functions with respect to \(\pi\), each Gibbs step defines an orthogonal projection \(P_F\) onto some intersection \(F\) of the closed subspaces \(E_i\), where \(i = 1, \ldots, K\), of functions that are constant with respect to the \(i\)th coordinate. Combinations of these Gibbs steps are formed through a \emph{cycle} or a \emph{mixture}: a product or convex combination, respectively, of Gibbs steps such that the resulting sampling algorithm updates every coordinate with positive probability. Consider the following examples.

Two widely implemented Gibbs samplers are the \emph{deterministic-scan Gibbs sampler} and the \emph{random-scan Gibbs sampler}. The deterministic-scan Gibbs sampler updates each coordinate one at a time and has operator
\begin{equation*}
	P_{E_{\rho(1)}} \cdots P_{E_{\rho(K)}},
\end{equation*}
where \(\rho\) is some permutation of \(\{1, \ldots, K\}\). The random-scan Gibbs sampler randomly chooses a coordinate to update in each iteration and has operator 
\begin{equation*}
	\sum_{i = 1}^K w_i P_{E_i},
\end{equation*}
where \(w_1, \ldots, w_K\) are positive weights satisfying \(\sum_{i = 1}^K w_i = 1\). The sampling algorithms of these two samplers are delineated in Algorithms~\ref{alg: deterministic-scan Gibbs sampler corresponding to a permutation rho} and \ref{alg: random-scan Gibbs sampler with weights w}. Collectively, these two types of Gibbs samplers are sometimes called \emph{full} Gibbs samplers.

\begin{algorithm}[H]
	\caption{Deterministic-scan Gibbs sampler corresponding to a permutation \(\rho\) of \(\{1, \ldots, K\}\)} 
	\label{alg: deterministic-scan Gibbs sampler corresponding to a permutation rho}
	\begin{algorithmic}[1]
		\setcounter{ALG@line}{-1}
		\State {\it Input:} Current value \(X^{(n)} = (x_{\rho(1)}^{(n)}, \ldots, x_{\rho(K)}^{(n)})\).
		\State Draw \(X^{(n+1)}_{\rho(1)} \sim \pi(\cdot|x_{\rho(2)}^{(n)}, \ldots, x_{\rho(K)}^{(n)})\), and call the observed value \(x_{\rho(1)}^{(n+1)}\).
		\State Draw \(X^{(n+1)}_{\rho(2)} \sim \pi(\cdot|x_{\rho(1)}^{(n+1)}, x_{\rho(3)}^{(n)}, \ldots, x_{\rho(K)}^{(n)})\), and call the observed value \(x_{\rho(2)}^{(n+1)}\). 
		
		\quad \vdots
		\setcounter{ALG@line}{10}
		{\let\fmtlinenumber\mathalph 
			\State Draw \(X^{(n+1)}_{\rho(K)} \sim \pi(\cdot|x_{\rho(1)}^{(n+1)}, \ldots, x_{\rho(K-1)}^{(n+1)})\), and call the observed value \(x_{\rho(K)}^{(n+1)}\).
		}
		
		\setcounter{ALG@line}{10}
		\renewcommand\mathalph[1]{$\Alph{#1} + 1$}
		{\let\fmtlinenumber\mathalph 
			\State Set \(n = n+1\).
		}
	\end{algorithmic}
\end{algorithm}

\begin{algorithm}[H]
	\caption{Random-scan Gibbs sampler with weights \(w = (w_1, \ldots, w_K)\)} 
	\label{alg: random-scan Gibbs sampler with weights w}
	\begin{algorithmic}[1]
		\setcounter{ALG@line}{-1}
		\State {\it Input:} Current value \(X^{(n)} = (x_1^{(n)}, \ldots, x_K^{(n)})\).
		
		\State Draw a random variable \(\mscr{I}\) with \(\Pr(\mscr{I} = k) = w_k\) for \(k = 1, \ldots, K\), and call the observed value \(i\).
		
		\State Draw \(X^{(n+1)}_i \sim \pi(\cdot|x_1^{(n)}, \ldots, x_{i - 1}^{(n)}, x_{i+1}^{(n)}, \ldots x_K^{(n)})\), and call the observed value \(x_i^{(n+1)}\).
		
		\State Set \(X^{(n+1)} = (x_1^{(n)}, \ldots, x_{i-1}^{(n)}, x_i^{(n+1)}, x_{i+1}^{(n)}, \ldots, x_K^{(n)})\).
		
		\State Set \(n = n + 1\).
	\end{algorithmic}
\end{algorithm}

In contrast to full Gibbs samplers, \emph{blocked} Gibbs samplers jointly update two or more coordinates. For instance, a blocked Gibbs sampler with operator \(P_{E_1 \intersect E_2} P_{E_3} \cdots P_{E_K}\) has sampling algorithm given by Algorithm~\ref{alg: example of a blocked deterministic-scan Gibbs sampler}. Gibbs steps can also be combined so that a coordinate is updated more than once in one pass of the resulting algorithm; for example, if \(K = 3\), then the cycle \(P_{E_1}P_{E_2}P_{E_1}P_{E_2}P_{E_3}\) updates the first two coordinates twice in each iteration of its algorithm. We additionally consider samplers that target marginal distributions of \(\pi\) formed via cycles and mixtures of \emph{collapsed} Gibbs steps. 

\begin{algorithm}[H]
	\caption{A blocked deterministic-scan Gibbs sampler} 
	\label{alg: example of a blocked deterministic-scan Gibbs sampler}
	\begin{algorithmic}[1]
		\setcounter{ALG@line}{-1}
		\State {\it Input:} Current value \(X^{(n)} = (x_{1}^{(n)}, \ldots, x_{K}^{(n)})\).
		\State Draw \((X^{(n+1)}_1, X^{(n+1)}_2) \sim \pi(\cdot|x_3^{(n)}, \ldots, x_{K}^{(n)})\), and call the observed value \((x_{1}^{(n+1)}, x_2^{(n+1)})\).
		\State Draw \(X^{(n+1)}_{3} \sim \pi(\cdot|x_{1}^{(n+1)}, x_{2}^{(n+1)}, x_4^{(n)}, \ldots, x_{K}^{(n)})\), and call the observed value \(x_{3}^{(n+1)}\). 
		
		\quad \vdots
		\setcounter{ALG@line}{10}
		\renewcommand\mathalph[1]{$\Alph{#1} - 1$}
		{\let\fmtlinenumber\mathalph 
			\State Draw \(X^{(n+1)}_{K} \sim \pi(\cdot|x_{1}^{(n+1)}, \ldots, x_{K-1}^{(n+1)})\), and call the observed value \(x_{K}^{(n+1)}\).
		}
		
		\setcounter{ALG@line}{10}
		\renewcommand\mathalph[1]{$\Alph{#1}$}
		{\let\fmtlinenumber\mathalph 
			\State Set \(n = n+1\).
		}
	\end{algorithmic}
\end{algorithm}

While updating in blocks or collapsing down variables may not always be practically feasible, the general belief is that blocking and collapsing can help to remove unwanted autocorrelation in the Markov chain and speed up mixing (\cite{Liu:1994}, \cite{Liu:Wong:Kong:1994}). Beware, however, that such a strategy is not foolproof: examples of Gibbs samplers targeting Gaussian distributions where blocking results in slower convergence to stationarity are given in \cite{Liu:Wong:Kong:1994} and \cite{Roberts:Sahu:1997}.

In this paper, we establish relationships of a spectral nature between different cycles and mixtures of Gibbs steps. Generalizing the solidarity result of \cite{Chlebicka:Latuszynski:Miasojedow:2025} regarding full Gibbs samplers, we show that every cycle and mixture corresponding to the same choice of Gibbs steps has a spectral gap or none of them do. The solidarity result of \cite{Chlebicka:Latuszynski:Miasojedow:2025} is the special case between cycles and mixtures of \(P_{E_1}, \ldots, P_{E_K}\): the presence of a spectral gap, or lack thereof, is shared between all deterministic-scan Gibbs samplers with any choice of permutation and all random-scan Gibbs samplers with any choice of weights. Our generalized solidarity principle allows us to establish the following qualitative hierarchy between full Gibbs samplers and cycles and mixtures of Gibbs steps. 

\begin{thm}
	\label{thm: introduction, inheritence of a spectral gap}
	Every cycle and mixture of Gibbs steps, collapsed or otherwise, has a spectral gap whenever a full Gibbs sampler has a spectral gap.
\end{thm}

\begin{proof}
	The theorem follows from Theorem~\ref{thm: inheritance of a spectral gap from the full Gibbs sampler} and Corollary~\ref{cor: inheritance of a spectral gap for cycles and mixtures of collapsed Gibbs steps}.
\end{proof}

Theorem~\ref{thm: introduction, inheritence of a spectral gap} extends the solidarity property of the spectral gap enjoyed by full Gibbs samplers in the following way: none of the full Gibbs samplers have a spectral gap provided some blocked or collapsed variant fails to have a spectral gap. Furthermore, in the sense of having a spectral gap, Theorem~\ref{thm: introduction, inheritence of a spectral gap} shows that blocked and collapsed Gibbs samplers are never qualitatively worse than their full counterparts. Previous theoretical justification for this outlook by \cite{Liu:1994} and \cite{Liu:Wong:Kong:1994} made comparisons between the norms of the operators of a deterministic-scan Gibbs sampler and two corresponding blocked and collapsed Gibbs samplers. Using different techniques, our approach handles a larger class of Markov operators and yields more direct results about their spectral gaps. Having a spectral gap is valuable because, under very mild regularity conditions, the existence of a spectral gap implies geometric ergodicity.

We additionally establish exact relations between the spectra corresponding to blocked and collapsed variants of a Gibbs sampler. In particular, when there is conditional independence, we show that certain blocked Gibbs samplers share the same spectral gap with corresponding collapsed Gibbs samplers. We demonstrate how these results can simplify analysis of a blocked Gibbs sampler by examining a setting where one blocked Gibbs sampler has a spectral gap while another blocked Gibbs sampler does not. Although intuitively tempting, there is consequently no analogue of Theorem~\ref{thm: introduction, inheritence of a spectral gap} between blocked or collapsed Gibbs samplers: a blocked (resp. collapsed) Gibbs sampler does not in general inherit a spectral gap from another blocked (resp. collapsed) Gibbs sampler. In fact, our example shows that geometric ergodicity is not necessarily preserved when blocking or collapsing in different ways.

Previous work by \cite{Amit:1991}, \cite{Diaconis:Khare:Saloff-Coste:2010}, \cite{Rosenthal:Rosenthal:2015}, and \cite{Chlebicka:Latuszynski:Miasojedow:2025} also adopted the approach of studying Gibbs samplers in the \(L^2\)-setting as operators comprised of orthogonal projections. However, Gibbs samplers therein were not considered in tandem with their blocked or collapsed variants. 

For the special case of two-component Gibbs samplers, \cite{Qin:Jones:2022} established an exact quantitative relationship between the \(L^2\)-convergence rates of deterministic-scan Gibbs samplers and random-scan Gibbs samplers: the authors showed that the deterministic-scan Gibbs sampler is always faster. Their results leverage the connection between two-component deterministic-scan Gibbs samplers and their reversible marginal chains. When there are more than two components, the marginal chains of deterministic-scan Gibbs samplers are not necessarily reversible, so establishing similar quantitative results remains challenging. There is also no best choice in terms of convergence time between deterministic-scan Gibbs samplers and random-scan Gibbs samplers when there are more than two components: \cite{Roberts:Rosenthal:2015} and \cite{He:De_Sa:Mitliagkas:Re:2016} provide examples and some heuristic guidance for when one might be better than the other.

The rest of the paper is organized as follows. Section~\ref{section: Preliminaries} introduces the theoretical background concerning Markov operators given by Gibbs steps and their combinations. Section~\ref{section: A generalized solidarity principle of the spectral gap} establishes the generalized solidarity principle for cycles and mixtures of Gibbs steps. Section~\ref{section: Inheritance of a spectral gap from a full Gibbs sampler} explores different connections between the spectra corresponding to full Gibbs samplers and their blocked and collapsed variants. In particular, the results needed to prove Theorem~\ref{thm: introduction, inheritence of a spectral gap} are established in this section. Section~\ref{section: A spectral disconnection} provides the example showing that neither blocked nor collapsed Gibbs samplers in general share a solidarity property of the spectral gap. Some technical details are relegated to the appendices.

\section{Preliminaries}
\label{section: Preliminaries}

\subsection{Markov operators}

Fix an integer \(K \geq 2\), and for any positive integer \(n\), let \([n] = \{1, \ldots, n\}\). For \(k \in [K]\), let \(\mcal{X}_k\) be a complete separable metric space equipped with its Borel \(\sigma\)-algebra \(\mscr{B}(\mcal{X}_k)\). Further, let \(\mcal{X} = \mcal{X}_1 \cross \cdots \cross \mcal{X}_K\) denote the corresponding \(K\)-fold product space equipped with its Borel \(\sigma\)-algebra \(\mscr{B}(\mcal{X})\). Note that \(\mcal{X}\) is itself a complete separable metric space. Fix a Borel probability measure \(\pi\) on \((\mcal{X}, \mscr{B}(\mcal{X}))\), and let \(L^2(\pi)\) denote the Hilbert space given by the set of measurable functions \(f \colon \mcal{X} \to \mbb{C}\) satisfying \(\int_\mcal{X} \abs{f(x)}^2 \, \pi(dx) < \infty\). As usual, functions that agree \(\pi\)-a.e. are identified with each other, and the inner product on \(L^2(\pi)\) is given by \(\inner{f,g} = \int_\mcal{X} f(x) \overline{g(x)} \, \pi(dx)\). 

The space of bounded linear operators on \(L^2(\pi)\) is denoted by \(B(L^2(\pi))\). Define \(\Pi \in B(L^2(\pi))\) by \(\Pi f \equiv \pi f \coloneq \int_\mcal{X} f(x) \, \pi(dx)\) for \(f \in L^2(\pi)\). The operator \(\Pi\) maps every \(f \in L^2(\pi)\) to the constant function equal to \(\pi f\). Observe that \(\Pi\) is an orthogonal projection; that is, \(\Pi\) is idempotent and self-adjoint. Idempotency is clear, and self-adjointness holds since \(\inner{\Pi f, g} = (\pi f)(\pi \overline{g}) = \inner{f, \Pi g}\) for all \(f,g \in L^2(\pi)\).

The probability measure \(\pi\) is invariant (or stationary) for a Markov kernel \(Q(x,dy)\) defined on \(\mcal{X} \cross \mscr{B}(\mcal{X})\) whenever \(\pi(A) = \int_\mcal{X} Q(x, A) \, \pi(dx)\) for all \(A \in \mscr{B}(\mcal{X})\). (The introduction of other standard properties like irreducibility or aperiodicity is not needed until Section~\ref{section: A spectral disconnection}.) A Markov kernel \(Q(x,dy)\) with invariant probability measure \(\pi\) defines a corresponding Markov operator \(Q \in B(L^2(\pi))\) that is given by 
\begin{equation*}
	[Qf](x) = \int_{\mcal{X}} f(y) \, Q(x, dy), \quad x \in \mcal{X},
\end{equation*}
for \(f \in L^2(\pi)\). By the invariance of \(\pi\) for \(Q(x,dy)\), the operator \(Q\) satisfies \(\Pi Q = \Pi\). Two Markov kernels \(Q_1(x,dy)\) and \(Q_2(x,dy)\) are identified with each other if the functions \(x \mapsto Q_1(x, A)\) and \(x \mapsto Q_2(x, A)\) agree \(\pi\)-a.e. for all \(A \in \mscr{B}(\mcal{X})\).

An operator \(Q \in B(L^2(\pi))\) is a \emph{Markov operator} if \(Q\mathbf{1} = \mathbf{1}\) (where \(\mathbf{1}\) is the constant function equal to 1) and \(Qf \geq 0\) whenever \(f \geq 0\). Any Markov operator \(Q\) satisfying \(\Pi Q = \Pi\) defines a corresponding Markov kernel that is given by \(Q(x,A) = (Q\mathbf{1}_A)(x)\) for all \((x,A) \in \mcal{X} \cross \mscr{B}(\mcal{X})\), where \(\mathbf{1}_A\) denotes the indicator function satisfying \(\mathbf{1}_A(x) = 1\) if \(x \in A\) and 0 otherwise. Since \(\Pi Q = \Pi\), the corresponding kernel \(Q(x,dy)\) has invariant probability measure \(\pi\). (See, e.g., \cite{Bakry:Gentil:Ledoux:2014} for details and a more general treatment of Markov operators.)

In view of this one-to-one correspondence between a Markov kernel \(Q(x,dy)\) with invariant probability measure \(\pi\) and a Markov operator \(Q\) satisfying \(\Pi Q = \Pi\), we sometimes refer to both the kernel and its corresponding operator by the same name for the sake of brevity. Whenever this is the case, it is clear from context which object is being referred to.

\subsection{Gibbs steps}

Fix any nonempty proper subset \(I\) of \([K]\), and write \(I = \{i_1, \ldots i_n\}\) and \([K] \setminus I = \{j_1, \ldots, j_{K - n}\}\), where \(i_1 \leq \cdots \leq i_n\) and \(j_1 \leq \cdots \leq j_{K - n}\).  Define \(\mcal{X}_I = \bigtimes_{\ell = 1}^n \mcal{X}_{i_\ell}\) and \(\mcal{X}_{-I} = \bigtimes_{\ell = 1}^{K - n} \mcal{X}_{j_\ell}\). For any \(x = (x_1, \ldots, x_K) \in \mcal{X}\), let \(x_I = (x_{i_1}, \ldots, x_{i_n})\) and \(x_{-I} = (x_{j_1}, \ldots, x_{j_{K - n}})\) denote the projections of \(x\) onto the coordinates in \(\mcal{X}_I\) and \(\mcal{X}_{-I}\), respectively. As a shorthand notation, the point \(x\) is sometimes denoted by \((x_I, x_{-I})\). When \(I = \{i\}\) for some \(i \in [K]\), denote \(\mcal{X}_i\) and \(\mcal{X}_{-i}\) in place of \(\mcal{X}_{\{i\}}\) and \(\mcal{X}_{-\{i\}}\), respectively, and denote \(x_i\) and \(x_{-i}\) in place of \(x_{\{i\}}\) and \(x_{-\{i\}}\), respectively. The probability measure \(\pi(\cdot|x_{-I})\) is the distribution of \(X_I | X_{-I} = x_{-I}\), where \((X_1, \ldots, X_K)\) is a random element taking values in \(\mcal{X}\) with distribution \(\pi\).

For every \(i \in [K]\), define \(E_i\) to be the closed (linear) subspace of functions in \(L^2(\pi)\) that are constant with respect to the \(i\)th coordinate; that is, \(f \in E_i\) if
\begin{equation*}
	f(x_1, \ldots, x_{i - 1}, x_i, x_{i+1}, \ldots, x_K) = f(x_1, \ldots, x_{i - 1}, x_i', x_{i+1}, \ldots, x_K) 
\end{equation*}
for all \(x_i, x_i' \in \mcal{X}_i\) and \(x_{-i} \in \mcal{X}_{-i}\). The closed subspace of constant functions is denoted by \(E \coloneq \bigcap_{i = 1}^K E_i\). Since orthogonal projections are uniquely determined by their range, the projection onto \(E\) is precisely the operator \(\Pi\). 

We call the orthogonal projection \(P_F \in B(L^2(\pi))\) onto \(F = \bigcap_{i \in I} E_i\) a \emph{Gibbs step}. Its rule can be written as
\begin{equation}
	\label{eq: P_F}
	(P_Ff)(x) = \int_{\mcal{X}_I} f(x_I', x_{-I}) \, \pi(dx_I'|x_{-I}), \quad x \in \mcal{X},
\end{equation}
for \(f \in L^2(\pi)\). Indeed, in view of \eqref{eq: P_F}, idempotency of \(P_F\) is immediate, and self-adjointness holds since
\begin{align*}
	&\inner{P_F f, g} \\
	&= \int_{\mcal{X}_I \cross \mcal{X}_{-I}} (P_F f)(x_I, x_{-I}) \, \overline{g(x_I, x_{-I})} \, \pi(dx_I, dx_{-I}) \\
	&\qquad= \int_{\mcal{X}_I \cross \mcal{X}_{-I}} \paren*{\int_{\mcal{X}_I} f(x_I', x_{-I}) \, \pi(dx_I'|x_{-I})} \overline{g(x_I, x_{-I})} \, \pi(dx_I, dx_{-I}) \\
	&\qquad= \int_{\mcal{X}_{-I}} \bracket*{ \int_{\mcal{X}_I}  \paren*{\int_{\mcal{X}_I} f(x_I', x_{-I}) \, \pi(dx_I'|x_{-I})} \overline{g(x_I, x_{-I})} \, \pi(dx_I|x_{-I})} \pi(dx_{-I}) \\
	&\qquad= \int_{\mcal{X}_{-I}} \bracket*{\int_{\mcal{X}_I} f(x_I', x_{-I}) \paren*{\int_{\mcal{X}_I} \overline{g(x_I, x_{-I})} \, \pi(dx_I|x_{-I})} \, \pi(dx_I'|x_{-I})} \pi(dx_{-I}) \\
	&\qquad= \int_{\mcal{X}_{-I}} \bracket*{\int_{\mcal{X}_I} f(x_I', x_{-I}) \overline{(P_F g)(x_I', x_{-I})} \, \pi(dx_I'|x_{-I})} \pi(dx_{-I})\\
	&\qquad= \int_{\mcal{X}_I \cross \mcal{X}_{-I}} f(x_I', x_{-I}) \overline{(P_Fg)(x_I', x_{-I})} \, \pi(dx_I', dx_{-I}) \\
	&\qquad= \inner{f, P_F g}
\end{align*} 
for all \(f,g \in L^2(\pi)\). Additionally, \(\Pi P_F = \Pi\) because 
\begin{equation*}
	\begin{aligned}[b]
		(\Pi P_F)(f) &\equiv \int_{\mcal{X}_I \cross \mcal{X}_{-I}} (P_Ff)(x_I, x_{-I}) \, \pi(dx_I, dx_{-I}) \\
		&= \int_{\mcal{X}_I \cross \mcal{X}_{-I}} \bracket*{\int_{\mcal{X}_I} f(x_I', x_{-I}) \, \pi(dx_I'|x_{-I})} \pi(dx_I, dx_{-I}) \\
		&= \int_{\mcal{X}_{-I}} \bracket*{\int_{\mcal{X}_I} f(x_I', x_{-I}) \, \pi(dx_I'|x_{-I})} \pi(dx_{-I}) \\
		&= \int_{\mcal{X}_I \cross \mcal{X}_{-I}} f(x_I', x_{-I}) \pi(dx_I', x_{-I}) \\
		&\equiv \Pi f
	\end{aligned}
\end{equation*}
for all \(f \in L^2(\pi)\).

Since \(P_F\) is a Markov operator satisfying \(\Pi P_F = \Pi\), it defines a Markov kernel given by \(P_{F}(x,A) = (P_{F}\mathbf{1}_A)(x)\) with invariant probability measure \(\pi\). The resulting sampling algorithm transitions from a given point \(x = (x_1, \ldots, x_K)\) by updating the coordinates whose indices lie in \(I\) via the conditional distribution \(\pi(\cdot|x_{-I})\). Since a single Gibbs step can never update every coordinate (as \(I\) is a proper subset of \([K]\)), multiple Gibbs steps must be combined to form a useful sampling algorithm. We discuss now how these combinations are formed.

\subsection{Cycles and mixtures of Gibbs steps}
\label{subsection: Cycles and mixtures of Gibbs steps}

Gibbs steps are the building blocks with which we construct samplers to target \(\pi\). Since each Gibbs step defines a kernel with invariant probability measure \(\pi\), so too do products and convex combinations of Gibbs steps. Inspired by the language in \cite{Tierney:1994}, we say these products and convex combinations are \emph{cycles} and \emph{mixtures} of Gibbs steps, respectively, provided the associated sampling algorithm has a positive probability of updating each coordinate. 

To be precise, for any integer \(g \geq 2\), let 
\begin{equation}
	\label{eq: S_g}
	\mscr{S}_g = \curly*{(w_1, \ldots, w_g) \in (0, 1)^g : \sum_{d = 1}^g w_d = 1}.
\end{equation} 
Taking any nonempty proper subsets \(I_1, \ldots, I_g\) of \([K]\) that satisfy \(\bigcup_{d = 1}^g I_d = [K]\) and letting \(F_d = \bigcap_{i \in I_d} E_i\) for \(d \in [g]\), we say an operator is a \emph{cycle} of Gibbs steps if it has the form \(P_{F_1} \cdots P_{F_g}\), and we say an operator is a \emph{mixture} of Gibbs steps if it has the form \(\sum_{d = 1}^g w_dP_{F_d}\) for some choice of weights \((w_1, \ldots, w_g) \in \mscr{S}_g\).

The following four types of Gibbs samplers are commonly employed and studied, so we make special mention of their Markov operators. 
\begin{enumerate}
	\item A \emph{deterministic-scan Gibbs sampler} has Markov operator 
	\begin{equation}
		\label{eq: deterministic-scan Gibbs sampler operator}
		P_{E_{\rho(1)}} \cdots P_{E_{\rho(K)}},
	\end{equation}
	where \(\rho\) is any partition of \([K]\).
	
	\item A \emph{blocked deterministic-scan Gibbs sampler} has Markov operator 
	\begin{equation}
		\label{eq: blocked deterministic-scan Gibbs sampler operator}
		P_{F_1}\cdots P_{F_g}, 
	\end{equation} 
	where \(\{I_1, \ldots, I_g\}\) forms a partition of \([K]\) with \(2 \leq g < K\) and \(F_d = \bigcap_{i \in I_d} E_i\) for \(d \in [g]\).   
	
	\item A \emph{random-scan Gibbs sampler} has Markov operator 
	\begin{equation}
		\label{eq: random-scan Gibbs sampler operator}
		\sum_{i = 1}^K w_i P_{E_i},
	\end{equation}
	where \((w_1, \ldots, w_K) \in \mscr{S}_K\).
	
	\item A \emph{blocked random-scan Gibbs sampler} has Markov operator
	\begin{equation}
		\label{eq: blocked random-scan Gibbs sampler operator}
		\sum_{d = 1}^g v_d P_{F_d},
	\end{equation}
	where \((v_1, \ldots, v_g) \in \mscr{S}_g\) and where \(\{I_1, \ldots, I_g\}\) forms a partition of \([K]\) with \(2 \leq g < K\) and \(F_d = \bigcap_{i \in I_d} E_i\) for \(d \in [g]\).
\end{enumerate}
As mentioned in the introduction, a Gibbs sampler is sometimes called \emph{full} when its operator is given by \eqref{eq: deterministic-scan Gibbs sampler operator} or \eqref{eq: random-scan Gibbs sampler operator}. 

Put more simply, Gibbs samplers are special cases of cycles and mixtures of Gibbs steps where the sets \(I_1,\ldots,I_g\) 
are chosen to produce a partition of \([K]\). No such restriction need be made in general. For instance, if there is some coordinate \(i\) such that draws from \(\pi(\cdot|x_{-i})\) for \(x_{-i} \in \mcal{X}_{-i}\) incur less computation cost compared to draws from other conditional distributions, then a mixture of Gibbs steps where coordinate \(i\) is updated multiple times in one pass of the algorithm may be appealing.

Cycles of Gibbs steps are also closely related to the \emph{partially collapsed Gibbs samplers} introduced in \cite{van_Dyk:Park:2008}. Certain cycles of Gibbs steps result in an algorithm where some intermediate quantities are not conditioned on. Trimming these quantities results in a sampling algorithm for a partially collapsed Gibbs sampler. This particular connection is not pursued any further in this paper.

\section{A generalized solidarity principle of the spectral gap}
\label{section: A generalized solidarity principle of the spectral gap}

Let \(H\) be any complex (nonzero) Hilbert space. Let \(B(H)\) denote the algebra of bounded linear operators \(T \colon H \to H\) with norm given by \(\norm{T} = \inf\{C \geq 0 : \norm{Th} \leq C \norm{h} \text{ for all } h \in H\}\). For any \(T \in B(H)\), let \(\sigma(T) = \{\lambda \in \mbb{C} : T - \lambda \, \Id \text{ is not invertible}\}\) denote the spectrum of \(T\), where \(\Id\) denotes the identity operator, and let \(r(T) = \sup\{\abs{\lambda} : \lambda \in \sigma(T)\}\) denote the spectral radius of \(T\). For any closed subspace \(S\), denote the orthogonal projection of \(H\) onto \(S\) by \(P_S\) and the orthogonal complement of \(S\) by \(S^\perp\). Recall that \(r(T) = \norm{T}\) for any self-adjoint operator \(T \in B(H)\). 

We say a Markov operator \(Q\) satisfying \(\Pi Q = \Pi\) has a \emph{spectral gap} if \(r(Q - \Pi) < 1\), in which case the gap is given by \(1 - r(Q - \Pi)\). (See Remark~\ref{rem: equivalent definition of spectral gap} for an equivalent definition.) Observe that \(Q^n - \Pi = (Q - \Pi)^n\) for all \(n \geq 1\) since \(\Pi Q = \Pi = Q \Pi\) and \(\Pi^2 = \Pi\). Roughly put, the larger the spectral gap of \(Q\), the faster the sequence \((Q^n)_{n \geq 1}\) converges in norm to \(\Pi\). When \(Q\) is self-adjoint so that \(\norm{(Q - \Pi)^n} = \norm{Q - \Pi}^n\) for all \(n \geq 1\), this relation holds in an exact sense since 
\begin{equation*}
	\norm{Q^n - \Pi} = \norm{(Q - \Pi)^n} = \norm{Q - \Pi}^n = r(Q - \Pi)^n, \quad n \geq 1.
\end{equation*} 
Even when \(Q\) is not necessarily self-adjoint, the relation holds in an asymptotic sense since 
\begin{equation*}
	\lim_{n \to \infty} \norm{Q^n - \Pi}^{1/n} = r(Q - \Pi)
\end{equation*} 
by the spectral radius formula (see, e.g., \cite{Murphy:1990}, Theorem~1.2.7). Under mild regularity conditions, a spectral gap also implies geometric ergodicity: exponentially fast convergence to stationarity in total variation distance. (Geometric ergodicity is reviewed more precisely in Section~\ref{subsection: Background on geometric ergodicity}.)  

By the spectral radius formula, \(Q\) has a spectral gap if and only if the sequence \((Q^n)_{n \geq 1}\) converges in norm to \(\Pi\). Hence, to understand when cycles and mixtures of Gibbs steps have a spectral gap, one may equivalently pose the question of when norm convergence occurs. This question has been studied in the following broader framework: given closed proper subspaces \(M_1, \ldots, M_N\) with \(M \coloneq \bigcap_{i = 1}^N M_i\), one may ask when do the sequences \(((P_{M_1}\cdots P_{M_N})^n)_{n \geq 1}\) and \(((\sum_{i = 1}^N w_i P_{M_i})^n)_{n \geq 1}\), where \((w_1, \ldots, w_N) \in \mscr{S}_N\), converge in norm to \(P_M\). For sequences generated by the powers of a product of projections, \cite{Badea:Grivaux:Muller:2011} established a list of equivalent conditions for norm convergence. In the following lemma, only the portions of Theorem~4.1 of \cite{Badea:Grivaux:Muller:2011} that are relevant to our purposes are stated.

\begin{lem}
	\label{lem: characterization of ||P_(M_1) ... P_(M_N) - P_M|| < 1}
	Let \(N \geq 2\). Let \(M_1, \ldots, M_N\) be closed proper subspaces of \(H\) with intersection \(M = \bigcap_{i = 1}^N M_i\). The following assertions are equivalent:
	\begin{enumerate}
		\item \(1 \notin \sigma(P_{M_1} \cdots P_{M_N} - P_M)\);
		
		\item \(\norm{P_{M_1}\cdots P_{M_N} - P_M} < 1\);
		
		\item \(\norm{N^{-1} \sum_{i = 1}^N P_{M_i} - P_M} < 1\).
	\end{enumerate}
\end{lem}

For any integer \(g \geq 2\), let \(\mscr{P}_g\) denote the set of all permutations of \([g]\). Since \(N^{-1}\sum_{i = 1}^N P_{M_i} = N^{-1}\sum_{i = 1}^N P_{M_{\rho(i)}}\) for all \(\rho \in \mscr{P}_N\), an immediate consequence of Lemma~\ref{lem: characterization of ||P_(M_1) ... P_(M_N) - P_M|| < 1} is that \(\norm{P_{M_{\rho(1)}} \cdots P_{M_{\rho(N)}} - P_M} < 1\) for some permutation \(\rho \in \mscr{P}_N\) if and only if \(\norm{P_{M_{\rho'(1)}} \cdots P_{M_{\rho'(N)}} - P_M} < 1\) for every permutation \(\rho' \in \mscr{P}_N\). The following lemma shows that a similar statement can be made for convex combinations of \(P_{M_1}, \ldots, P_{M_N}\). The argument is essentially that used to prove Lemma~3.3 of \cite{Chlebicka:Latuszynski:Miasojedow:2025} (and Proposition~2 of \cite{Jones:Roberts:Rosenthal:2014}) and is provided in Appendix~\ref{section: proof of lemma for norms related to convex combinations of projections} for completeness.

\begin{lem}
	\label{lem: P in P_sum has norm less than 1 for some choice of weights implies norm less than 1 for any choice of weights}
	Let \(N \geq 2\). Let \(M_1, \ldots, M_N\) be closed proper subspaces of \(H\) with intersection \(M = \bigcap_{i = 1}^N M_i\). The following assertions are equivalent:
	\begin{enumerate}
		\item \(\norm*{\sum_{i = 1}^N w_i P_{M_i} - P_M} < 1\) for some \((w_1, \ldots, w_N) \in \mscr{S}_N\);
		\item \(\norm*{\sum_{i = 1}^N w_i' P_{M_i} - P_M} < 1\) for all \((w_1', \ldots, w_N') \in \mscr{S}_N\). 	
	\end{enumerate}
\end{lem}

As a direct consequence of Lemmas~\ref{lem: characterization of ||P_(M_1) ... P_(M_N) - P_M|| < 1} and \ref{lem: P in P_sum has norm less than 1 for some choice of weights implies norm less than 1 for any choice of weights}, the following theorem shows that every sequence generated by the powers of a product or convex combination of \(P_{M_1}, \ldots, P_{M_N}\) converges to \(P_M\) exactly whenever one such sequence converges to \(P_M\). 

\begin{thm}
	\label{thm: equivalent conditions of norm convergence}
	Let \(N \geq 2\). Let \(M_1, \ldots, M_N\) be closed proper subspaces of \(H\) with intersection \(M = \bigcap_{i = 1}^N M_i\). The following assertions are equivalent:
	\begin{enumerate}
		\item There exists a permutation \(\rho \in \mscr{P}_N\) or choice of weights \((w_1, \ldots, w_N) \in \mscr{S}_N\) such that
		\begin{align*}
			1 \notin \sigma(P_{M_{\rho(1)}} \cdots P_{M_{\rho(N)}} - P_M) &\text{ or } \norm{P_{M_{\rho(1)}} \cdots P_{M_{\rho(N)}} - P_M} < 1 \\
			&\text{ or } \norm*{\sum_{i = 1}^N w_i P_{M_i} - P_M} < 1.
		\end{align*}
		
		\item Every permutation \(\rho' \in \mscr{P}_N\) and every choice of weights \((w_1', \ldots, w_N') \in \mscr{S}_N\) satisfies 
		\begin{align*}
			1 \notin \sigma(P_{M_{\rho'(1)}} \cdots P_{M_{\rho'(N)}} - P_M) &\text{ and } \norm{P_{M_{\rho'(1)}} \cdots P_{M_{\rho'(N)}} - P_M} < 1 \\
			&\text{ and } \norm*{\sum_{i = 1}^N w_i' P_{M_i} - P_M} < 1.
		\end{align*}
	\end{enumerate}  
\end{thm}

\begin{rem}
	Theorem~\ref{thm: equivalent conditions of norm convergence} holds for any complex Hilbert space (not just \(L^2(\pi)\)). Moreover, as we now explain, it provides a general solidarity result concerning the rate of convergence of a class (of sequences) of operators that all converge pointwise to \(P_M\). 
	
	For any \(\rho \in \mscr{P}_N\) and any \(w = (w_1, \ldots, w_N) \in \mscr{S}_N\), let \(T_\rho = P_{M_{\rho(1)}} \cdots P_{M_{\rho(n)}}\) and \(T_w = \sum_{i=1}^N w_i P_{M_i}\). For all \(h \in H\), it is well-known that \(\lim_{n \to \infty} T_\rho^n h = P_Mh\) (\cite{von_Neumann:1949}, \cite{Halperin:1962}) and \(\lim_{n \to \infty} T_w^n h = P_M h\) (\cite{Lapidus:1981}, \cite{Reich:1983}). The resulting algorithms are called the ``method of alternating projections'' and the ``method of simultaneous projections'', respectively, and they provide a way to find the best approximation to a point in \(H\) from \(M\). See \cite{Deustch:1992} for a survey on the many applications of the former algorithm in different areas of mathematics (beyond the MCMC setting). 
	
	Let \(T\) denote \(T_\rho\) or \(T_w\). There is a dichotomy regarding the rate of convergence in the aforementioned two algorithms (\cite{Deutsch:Hundal:2010}): either the convergence of \((T^n)_{n \geq 1}\) to \(P_M\) is exponentially fast, that is, there exist a constant \(\rho \in [0,1)\) and a function \(C : H \to (0, \infty)\) such that 
	\begin{equation*}
		\norm*{T^nh - P_M h} \leq C(h) \rho^n, \quad n \geq 1,
	\end{equation*} 
	for all \(h \in H\) or the convergence is arbitrarily slow. Various characterizations of arbitrarily slow convergence are given in \cite{Badea:Grivaux:Muller:2011}. Theorem~\ref{thm: equivalent conditions of norm convergence} establishes that this dichotomy is in fact that of a larger one concerning all the operators in \(\mscr{T} \coloneq \{T_\rho : \rho \in \mscr{P}_N\} \union \{T_w : w \in \mscr{S}_N\}\): either the convergence of every sequence generated by the powers of an operator in \(\mscr{T}\) is exponentially fast or the convergence of every sequence generated by the powers of an operator in \(\mscr{T}\) is arbitrarily slow. 
\end{rem}

We now return to the question of when cycles and mixtures of Gibbs steps have a spectral gap. Recall that since any mixture of Gibbs steps \(Q_{\mathrm{mix}}\) is self-adjoint, \(Q_{\mathrm{mix}}\) has a spectral gap if and only if \(\norm{Q_{\mathrm{mix}} - \Pi} < 1\). Theorem~\ref{thm: equivalent conditions of norm convergence} shows that such an equivalence also holds for cycles of Gibbs steps and, moreover, that all these conditions must hold simultaneously for all cycles and mixtures corresponding to the same choice of Gibbs steps. 

\begin{cor}
	\label{cor: generalized solidarity principle of the spectral gap}
	Let \(I_1, \ldots, I_g\) be nonempty proper subsets of \([K]\) satisfying \(\bigcup_{d = 1}^g I_d = [K]\), and let \(F_d = \bigcap_{i \in I_d} E_i\) for \(d \in [g]\). Then the following assertions are equivalent: 
	\begin{enumerate}
		\item \label{item: exists a permutation for a spectral gap} There exists a permutation \(\rho \in \mscr{P}_g\) such that \(P_{F_{\rho(1)}}\cdots P_{F_{\rho(g)}}\) has a spectral gap.
		
		\item \label{item: exists weights for a spectral gap} There exists a choice of weights \((w_1, \ldots, w_g) \in \mscr{S}_g\) such that \(\sum_{d = 1}^g w_d P_{F_d}\) has a spectral gap.
		
		\item \label{item: exists a permutation for norm < 1} There exists a permutation \(\rho \in \mscr{P}_g\) such that \(\norm{P_{F_{\rho(1)}} \cdots P_{F_{\rho(g)}} - \Pi} < 1\).
		
		\item \label{item: for all permutations, a spectral gap} For every permutation \(\rho' \in \mscr{P}_g\), the cycle \(P_{F_{\rho'(1)}}\cdots P_{F_{\rho'(g)}}\) has a spectral gap.
		
		\item \label{item: for all weights, a spectral gap} For every choice of weights \((w_1', \ldots, w_N') \in \mscr{S}_g\), the mixture \(\sum_{d = 1}^g w_d' P_{F_d}\) has a spectral gap.
		
		\item \label{item: for all permutations, norm < 1} For every permutation \(\rho' \in \mscr{P}_g\), it holds that \(\norm{P_{F_{\rho'(1)}} \cdots P_{F_{\rho'(g)}} - \Pi} < 1\).
	\end{enumerate}
\end{cor} 

\begin{proof}
	By definition, \(\bigcap_{d = 1}^g F_d = E\) and \(P_E = \Pi\). If there exists a permutation \(\rho \in \mscr{P}_g\) such that \(r(P_{F_{\rho(1)}} \cdots P_{F_{\rho(g)}} - \Pi) < 1\), then it holds that \(1 \notin \sigma(P_{F_{\rho(1)}} \cdots P_{F_{\rho(g)}} - \Pi)\) which occurs if and only if \(\norm{P_{F_{\rho(1)}}\cdots P_{F_{\rho(g)}} - \Pi} < 1\) by Theorem~\ref{thm: equivalent conditions of norm convergence}. The converse is immediate since \(r(T) \leq \norm{T}\) for any \(T \in B(H)\). Hence, assertions \ref{item: exists a permutation for a spectral gap} and \ref{item: exists a permutation for norm < 1} are equivalent. Since, by self-adjointness, \(r(\sum_{d = 1}^g w_d P_{F_d} - \Pi) = \norm{\sum_{d = 1}^g w_d P_{F_d} - \Pi}\) for any \((w_1, \ldots, w_g) \in \mscr{S}_g\), assertions \ref{item: exists weights for a spectral gap} and \ref{item: exists a permutation for norm < 1} are equivalent by Theorem~\ref{thm: equivalent conditions of norm convergence}. 
	
	This establishes the equivalence of assertions \ref{item: exists a permutation for a spectral gap}, \ref{item: exists weights for a spectral gap}, and \ref{item: exists a permutation for norm < 1}. A similar argument establishes the equivalence of assertions \ref{item: for all permutations, a spectral gap}, \ref{item: for all weights, a spectral gap}, and \ref{item: for all permutations, norm < 1}. A final application of Theorem~\ref{thm: equivalent conditions of norm convergence} implies the equivalence of assertions \ref{item: exists a permutation for norm < 1} and \ref{item: for all permutations, norm < 1}. This completes the proof.
\end{proof}

Corollary~\ref{cor: generalized solidarity principle of the spectral gap} generalizes the solidarity principle of the spectral gap introduced in Theorem~1.1 of \cite{Chlebicka:Latuszynski:Miasojedow:2025}. Indeed, the solidarity result of \cite{Chlebicka:Latuszynski:Miasojedow:2025} is the special case of Corollary~\ref{cor: generalized solidarity principle of the spectral gap} where \(g = K\) and \(F_d = E_d\) for \(d \in [g]\). Furthermore, the proof of Corollary~\ref{cor: generalized solidarity principle of the spectral gap} is simpler than the one provided therein: there is no need to employ geometric quantities such as the Friedrichs angle or inclination as studied in \cite{Badea:Grivaux:Muller:2011}.

\section{Inheritance of a spectral gap from a full Gibbs sampler}
\label{section: Inheritance of a spectral gap from a full Gibbs sampler}

\subsection{Inheritance for cycles and mixtures of Gibbs steps}

Beyond the generalized solidarity principle of the spectral gap for cycles and mixtures of Gibbs steps established in Corollary~\ref{cor: generalized solidarity principle of the spectral gap}, we consider a situation where the presence of a spectral gap for a family of cycles and mixtures corresponding to one choice of Gibbs steps affects the presence of a spectral gap for another family of cycles and mixtures corresponding to a different choice of Gibbs steps. This is motivated by the following known phenomenon: when a full Gibbs sampler (whose operator is given by \eqref{eq: deterministic-scan Gibbs sampler operator} or \eqref{eq: random-scan Gibbs sampler operator}) has a spectral gap, a blocked Gibbs sampler may not necessarily have a larger spectral gap. For example, if \(K = 3\), then it may be that \(r(P_{E_1}P_{E_2}P_{E_3} - \Pi) < r(P_{E_1}P_{E_2 \intersect E_3} - \Pi)\). See \cite{Liu:Wong:Kong:1994}, Example~3, for an example of a Gibbs sampler targeting a trivariate normal distribution where this is the case. This phenomenon raises the concern of the possibility of an extreme situation: one in which a full Gibbs sampler has a spectral gap while some blocked version does not. As we establish in the following theorem, such a situation can never occur. Every cycle and mixture of Gibbs steps inherits a spectral gap from a full Gibbs sampler. 

\begin{thm}
	\label{thm: inheritance of a spectral gap from the full Gibbs sampler}
	If a full Gibbs sampler has a spectral gap (or, equivalently, if every full Gibbs sampler has a spectral gap), then every cycle and mixture of Gibbs steps has a spectral gap.
\end{thm}

The following lemma is used in the proof of Theorem~\ref{thm: inheritance of a spectral gap from the full Gibbs sampler}.

\begin{lem}
	\label{lem: norm(P_M T) <= norm(P_(M_1) ... P_(M_N) T)}
	Let \(N \geq 2\). Let \(M_1, \ldots, M_N\) be closed subspaces of a complex Hilbert space \(H\) with intersection \(M = \bigcap_{i = 1}^N M_i\). For all \(T \in B(H)\),
	\begin{equation*}
		\norm*{P_M T} \leq \norm*{P_{M_1}\cdots P_{M_N}T}.
	\end{equation*}
\end{lem}

\begin{proof}
	The claim is trivial when \(M = \{0\}\), so suppose \(M \neq \{0\}\) in which case \(\norm{P_M} = 1\). Observe that \(P_M = P_{M_N}\cdots P_{M_1}P_M\), so \(P_M = P_MP_{M_1}\cdots P_{M_N}\) by taking adjoints. For all \(h \in H\), 
	\begin{align*}
		\norm{P_MTh} &= \norm{P_MP_{M_1}\cdots P_{M_N}Th} \\
		&\leq \norm{P_M P_{M_1}\cdots P_{M_N}T} \norm{h} \\
		&\leq \norm{P_{M_1}\cdots P_{M_N}T}\norm{h},
	\end{align*}
	where the second inequality follows by the submultiplicativity of the operator norm. This completes the proof.
\end{proof}

\begin{proof}[Proof of Theorem~\ref{thm: inheritance of a spectral gap from the full Gibbs sampler}]
	Take any choice of nonempty proper subsets \(I_1, \ldots, I_g\) of \([K]\) satisfying \(\bigcup_{d = 1}^g I_d = [K]\). For \(d \in [g]\), let \(F_d = \bigcap_{i \in I_d} E_i\) and let \(T_d = \prod_{i \in I_d} P_{E_i}\). By Corollary~\ref{cor: generalized solidarity principle of the spectral gap}, it suffices to show that \(\norm{P_{F_1} \cdots P_{F_g} - \Pi} < 1\). By Lemma~\ref{lem: norm(P_M T) <= norm(P_(M_1) ... P_(M_N) T)}, 
	\begin{align*}
		\norm*{P_{F_1}\cdots P_{F_g} - \Pi} &= \norm*{P_{F_1}\paren*{P_{F_2} \cdots P_{F_g} - \Pi}} \\
		&\leq \norm*{T_1 \paren*{P_{F_2} \cdots P_{F_g} - \Pi}} \\
		&= \norm*{T_1 P_{F_2} \cdots P_{F_g} - \Pi}.
	\end{align*}
	Hence, \(\norm{P_{F_1} \cdots P_{F_g} - \Pi} < 1\) if \( \norm*{T_1 P_{F_2} \cdots P_{F_g} - \Pi} < 1\). Theorem~\ref{thm: equivalent conditions of norm convergence} implies that \(\norm{T_1 P_{F_2} \cdots P_{F_g}- \Pi} < 1\) if and only if \(\norm{P_{F_2} T_1 P_{F_3} \cdots P_{F_g} - \Pi} < 1\). An appeal to Lemma~\ref{lem: norm(P_M T) <= norm(P_(M_1) ... P_(M_N) T)} once more yields that
	\begin{align*}
		\norm*{P_{F_2} T_1 P_{F_3} \cdots P_{F_g} - \Pi} &= \norm*{P_{F_2}\paren*{T_1 P_{F_3} \cdots P_{F_g} - \Pi}} \\
		&\leq \norm*{T_2 \paren*{T_1P_{F_3} \cdots P_{F_g} - \Pi}} \\
		&= \norm*{T_2 T_1 P_{F_3} \cdots P_{F_g} - \Pi}.
	\end{align*}
	Hence, \(\norm{P_{F_1} \cdots P_{F_g} - \Pi} < 1\) if \( \norm*{T_2 T_1 P_{F_3} \cdots P_{F_g} - \Pi} < 1\). By continuing in this way, it follows that \(\norm{P_{F_1} \cdots P_{F_g} - \Pi} < 1\) if \(\norm{T_g \cdots T_1 - \Pi} < 1\). Observe that \(T_g \cdots T_1\) is a product composed of \(P_{E_1}, \ldots, P_{E_K}\) wherein some of these projections may appear multiple times. Therefore, with a permutation that puts the same projections next to each other, it follows by Theorem~\ref{thm: equivalent conditions of norm convergence} that \(\norm{T_g \cdots T_1 - \Pi} < 1\) if and only if \(\norm{P_{E_1}\cdots P_{E_K} - \Pi} < 1\), which holds by assumption. This establishes the theorem.
\end{proof}

In the discussion at the beginning of Section~\ref{section: Inheritance of a spectral gap from a full Gibbs sampler}, we observed that \(P_{F_1} \cdots P_{F_g}\) and \(P_{F'_1} \cdots P_{F'_g}\), two cycles of Gibbs steps, in general may not satisfy
\begin{equation*}
	r(P_{F_1} \cdots P_{F_g} - \Pi) \leq r(P_{F'_1} \cdots P_{F'_g} - \Pi)
\end{equation*}  
even if \(F_d \subseteq F_d'\) for all \(d \in [g]\). (Recall the example in \cite{Liu:Wong:Kong:1994} concerning \(P_{E_1}P_{E_2 \intersect E_3} = P_{E_1}P_{E_2 \intersect E_3}P_{E_2 \intersect E_3}\) and \(P_{E_1}P_{E_2}P_{E_3}\).) However, we show in the following result that there is such an ordering between the spectral gaps of \(\sum_{d = 1}^g w_d P_{F_d}\) and \(\sum_{d = 1}^g w_d P_{F_d'}\), two mixtures of Gibbs steps, whenever \(F_d \subseteq F_d'\) for all \(d \in [g]\). This provides a choice of weights to construct blocked versions of mixtures of Gibbs steps that have a spectral gap at least as large as the original mixture. For instance, \(r((1/4)P_{E_1} + (1/2) P_{E_2 \intersect E_3} + (1/4) P_{E_4} - \Pi) \leq r(\sum_{i = 1}^4 (1/4) P_{E_i} - \Pi)\) when \(K = 4\). 

\begin{prop}
	\label{prop: inequality between spectral radii of two mixtures of Gibbs steps}
	Let \(\sum_{d = 1}^g w_d P_{F_d}\) and \(\sum_{d = 1}^g w_d P_{F_d'}\) be two mixtures of Gibbs steps with \(F_d \subseteq F_d'\) for every \(d \in [g]\). Then,
	\begin{equation}
		\label{eq: inequality between spectral radii of two mixtures of Gibbs steps}
		r\paren*{\sum_{d = 1}^g w_d P_{F_d} - \Pi} \leq r\paren*{\sum_{d = 1}^g w_d P_{F_d'} - \Pi}.
	\end{equation}  
\end{prop}

\begin{proof}
	Recall the following facts about positive operators on a complex Hilbert space \(H\). For any two self-adjoint operators \(T_1, T_2 \in B(H)\), where \(H\) is a complex Hilbert space, \(T_2 - T_1\) is positive if \(\inner{(T_2 - T_1)x, x} \geq 0\) for all \(x \in H\), in which case we write \(T_1 \leq T_2\). Linear combinations with positive coefficients of positive operators are positive. If \(0 \leq T_1 \leq T_2\), then \(\norm{T_1} \leq \norm{T_2}\). For any two closed subspaces \(S_1\) and \(S_2\) of \(H\), it holds that \(P_{S_1} \leq P_{S_2}\) if and only if \(S_1 \subseteq S_2\). (See, e.g., \cite{Murphy:1990}, Ch.~2, for more details.)
	
	Since \(E \subseteq F_d\) and \(E \subseteq F_d'\) for every \(d \in [g]\), the operators \(\sum_{d = 1}^g w_d P_{F_d} - \Pi = \sum_{d = 1}^g w_d(P_{F_d} - \Pi)\) and \(\sum_{d = 1}^g w_d P_{F_d'} - \Pi = \sum_{d = 1}^g w_d(P_{F_d'} - \Pi)\) are positive. Further, since \(F_d \subseteq F_d'\) for every \(d \in [g]\), it follows that
	\begin{equation*}
		0 \leq \sum_{d=1}^g w_d P_{F_d} - \Pi \leq \sum_{d = 1}^g w_d P_{F_d'} - \Pi,
	\end{equation*}
	so the respective spectral radii, which are equal to the respective norms by self-adjointness, are ordered in the same direction. This yields \eqref{eq: inequality between spectral radii of two mixtures of Gibbs steps}.
\end{proof}

\begin{rem}
	Under the conditions of Proposition~\ref{prop: inequality between spectral radii of two mixtures of Gibbs steps}, the Markov chain with operator \(\sum_{d = 1}^g w_d P_{F_d}\) is ``at least as efficient" as the Markov chain with operator \(\sum_{d = 1}^g w_d P_{F_d'}\) in the following sense: for any \(f \in L^2(\pi)\), the former chain produces estimates of \(\pi f\) that are asymptotically more efficient per iteration than estimates produced by the latter chain. See \cite{Mira:Geyer:1999} for more on this efficiency ordering.
\end{rem}

Since cycles and mixtures of Gibbs steps are central to our interests, our attention has been focused on \(L^2(\pi)\), but note that Theorem~\ref{thm: inheritance of a spectral gap from the full Gibbs sampler} (and Proposition~\ref{prop: inequality between spectral radii of two mixtures of Gibbs steps}) hold more generally for products and convex combinations of projections and their blocked variants on an arbitrary complex Hilbert space. In the next section, we consider properties of Gibbs steps that leverage the structure of \(L^2(\pi)\) more integrally.

\subsection{Connections to collapsed Gibbs steps}

Various \(L^2\) spaces are considered alongside \(L^2(\pi)\) throughout the rest of Section~\ref{section: Inheritance of a spectral gap from a full Gibbs sampler}. For any complete separable metric space \(\Omega\) and any probability measure \(\xi\) on \((\Omega, \mscr{B}(\Omega))\), let \(L^2(\xi)\) denote the Hilbert space given by the set of measurable functions \(f \colon \Omega \to \mbb{C}\) satisfying \(\int_{\Omega} \abs{f(\omega)}^2 \, \xi(d\omega) < \infty\). Functions that agree \(\xi\)-a.e. are identified with each other, and the inner product on \(L^2(\xi)\) is given by \(\inner{f,g} = \int_\Omega f(\omega) \overline{g(\omega)} \, \xi(d\omega)\). If \(Q \in B(L^2(\xi))\) is a Markov operator satisfying \(\Xi Q = \Xi\), where \(\Xi \in B(L^2(\xi))\) is given by \(\Xi f \equiv \xi f\), then \(Q\) is said to have a spectral gap whenever \(r(Q - \Xi) < 1\). 

Throughout this section, suppose \(K \geq 3\), and fix any nonempty proper subset \(I = \{i_1, \ldots, i_n\}\) of \([K]\) with \(n \geq 2\). Let \(\pi_{I}\) denote the marginal distribution of \(\pi\) on \(\mcal{X}_{I}\); that is, \(\pi_I(A_I) = \pi(A_I \cross \mcal{X}_{-I})\) for all \(A_I \in \mscr{B}(\mcal{X}_I)\). Let \(\Pi_I \in B(L^2(\pi_I))\) be given by \(\Pi_I f_I \equiv \pi_I f_I\) for \(f_I \in L^2(\pi_I)\). Further, for \(\ell \in [n]\), define \(E^{(I)}_{i_\ell}\) to be the closed subspace of functions in \(L^2(\pi_I)\) that are constant with respect to the \(i_\ell\)th coordinate. For any nonempty proper subset \(J\) of \(I\), we call the orthogonal projection \(P_{F_J^{(I)}} \in B(L^2(\pi_{I}))\) onto \(F_J^{(I)} = \bigcap_{j \in J} E_{j}^{(I)}\) a \emph{collapsed Gibbs step}. Cycles and mixtures of collapsed Gibbs steps are constructed analogously to those in Section~\ref{subsection: Cycles and mixtures of Gibbs steps} and yield samplers that target \(\pi_I\). Any cycle or mixture of collapsed Gibbs steps given by a form analogous to \eqref{eq: deterministic-scan Gibbs sampler operator}, \eqref{eq: blocked deterministic-scan Gibbs sampler operator}, \eqref{eq: random-scan Gibbs sampler operator}, or \eqref{eq: blocked random-scan Gibbs sampler operator} is called a \emph{collapsed Gibbs sampler}. 

Merely by taking the underlying space to be \(L^2(\pi_I)\) instead of \(L^2(\pi)\), observe that cycles and mixtures corresponding to the same choice of collapsed Gibbs steps follow a solidarity principle of the spectral gap (Corollary~\ref{cor: generalized solidarity principle of the spectral gap}) and that cycles and mixtures of collapsed Gibbs steps inherit a spectral gap from collapsed Gibbs samplers given by a form analogous to \eqref{eq: deterministic-scan Gibbs sampler operator} or \eqref{eq: random-scan Gibbs sampler operator} (Theorem~\ref{thm: inheritance of a spectral gap from the full Gibbs sampler}). Naturally, more substantial observations can be made when the space \(L^2(\pi)\) and the operators thereon are not completely ignored. We work to illustrate some of the relationships connecting cycles and mixtures of collapsed Gibbs steps with cycles and mixtures of Gibbs steps. These relationships are employed to develop an extension of Theorem~\ref{thm: inheritance of a spectral gap from the full Gibbs sampler}: cycles and mixtures of collapsed Gibbs steps inherit a spectral gap from a full Gibbs sampler. 

Let us introduce the following identification between \(L^2(\pi_I)\) and \(\bigcap_{i \in [K] \setminus I} E_i\).  

\begin{prop}
	\label{prop: L^2(pi_I) and bigcap_-I E_i are isometrically isomorphic}
	The linear mapping \(\Phi \colon L^2(\pi_I) \to \bigcap_{i \in [K] \setminus I} E_i\) given by
	\begin{equation*}
		(\Phi f_I)(x_I, x_{-I}) = f_I(x_I), \quad (x_I, x_{-I}) \in \mcal{X},
	\end{equation*}
	for \(f_I \in L^2(\pi_I)\) is an isometric isomorphism.
\end{prop}

\begin{proof}
	Fix any \(x_{-I}^* \in \mcal{X}_{-I}\). Surjectivity of \(\Phi\) is immediate, and injectivity follows from noting that two functions \(f,g \in \bigcap_{i \in [K] \setminus I} E_i\) agree if and only if \(f(x_I, x_{-I}^*) = g(x_I, x_{-I}^*)\) agree \(\pi_I\)-a.e. Hence, \(\Phi\) is an isomorphism with inverse given by 
	\begin{equation*}
		(\Phi^{-1}f)(x_I) = f(x_I, x_{-I}^*), \quad x_I \in \mcal{X}_I,
	\end{equation*} 
	for \(f \in \bigcap_{i \in [K] \setminus I} E_i\). Finally, for all \(f_I \in L^2(\pi_I)\), observe that
	\begin{align*}
		\norm{\Phi f_I}^2 &= \int_{\mcal{X}_I \cross \mcal{X}_{-I}} \abs{(\Phi f_I)(x_I, x_{-I})}^2 \, \pi(dx_I, dx_{-I}) \\
		&= \int_{\mcal{X}_I} \abs{(\Phi f_I)(x_I, x_{-I}^*)}^2 \, \pi(dx_I) \\
		&= \int_{\mcal{X}_I} \abs{f_I(x_I)}^2 \, \pi_I(dx_I) \\
		&= \norm{f_I}^2,
	\end{align*}
	so \(\Phi\) is isometric. This completes the proof.
\end{proof}

The following proposition shows that collapsed Gibbs steps on \(L^2(\pi_I)\) are \emph{similar} to certain Gibbs steps restricted to \(\bigcap_{i \in [K] \setminus I} E_i\). (Two bounded linear operators \(T_1\) and \(T_2\) on Hilbert spaces \(H_1\) and \(H_2\), respectively, are similar if there exists an invertible bounded linear operator \(\Psi \colon H_1 \to H_2\) such that \(T_2 = \Psi T_1 \Psi^{-1}\).)

\begin{prop}
	\label{prop: similarity between collapsed Gibbs steps and Gibbs steps with restricted domain}
	Let \(J\) be a nonempty subset of \(I\). Let \(F = \bigcap_{i \in [K] \setminus I} E_i\), let \(F_J = \bigcap_{j \in J} E_{j}\), and let \(F_J^{(I)} = \bigcap_{j \in J} E_{j}^{(I)}\). Then,
	\begin{equation*}
		(P_{F_J \intersect F})|_F = \Phi P_{F_J^{(I)}} \Phi^{-1},
	\end{equation*}
	where \(\Phi \colon L^2(\pi_I) \to F\) is the isometric isomorphism given in Proposition~\ref{prop: L^2(pi_I) and bigcap_-I E_i are isometrically isomorphic}.
\end{prop}

\begin{proof}
	First, suppose that \(J \neq I\). View \(\mcal{X}\) as the three-fold product \(\mcal{X}_{J} \cross \mcal{X}_{I \setminus J} \cross \mcal{X}_{-I}\), and use the shorthand notation \((x_J, x_{I-J}, x_{-I})\) to denote elements in \(\mcal{X}_{J} \cross \mcal{X}_{I \setminus J} \cross \mcal{X}_{-I}\). Fix any point \(x_{-I}^* \in \mcal{X}_{-I}\). For all \(f \in F\) and \((x_J, x_{I-J}, x_{-I}) \in \mcal{X}_{J} \cross \mcal{X}_{I \setminus J} \cross \mcal{X}_{-I}\), 
	\begin{align*}
		(P_{F_J \intersect F} f)(x_J, x_{I-J}, x_{-I}) &= \int_{\mcal{X}_J \cross \mcal{X}_{-I}} f(x_J', x_{I-J}, x_{-I}') \, \pi(dx_J', dx_{-I}'|x_{I-J}) \\
		&= \int_{\mcal{X}_J} f(x_J', x_{I-J}, x_{-I}^*) \, \pi(dx_J'|x_{I-J}) \\
		&= \int_{\mcal{X}_J} (\Phi^{-1}f)(x_J', x_{I-J}) \, \pi(dx_J'|x_{I-J}) \\
		&= (P_{F_J^{(I)}}\Phi^{-1}f)(x_J, x_{I - J}) \\
		&= (\Phi P_{F_J^{(I)}}\Phi^{-1}f)(x_J, x_{I - J}, x_{-I}).
	\end{align*}
	
	Now, suppose that \(J = I\). Then, \(P_{F_J \intersect F} = \Pi\) and \(P_{F_J^{(I)}} = \Pi_I\), and a similar argument on \(\mcal{X}_I \cross \mcal{X}_{-I}\) yields 
	\begin{equation*}
		\Pi f \equiv \int_{\mcal{X}_I} f(x_I, x_{-I}^*) \, \pi(dx_I) = \int_{\mcal{X}_I} (\Phi^{-1}f)(x_I) \, \pi(dx_I) \equiv \Phi \Pi_I \Phi^{-1} f
	\end{equation*}
	for all \(f \in F\). This completes the proof.
\end{proof}

For concreteness, suppose for the moment that \(K = 4\) and \(I = \{1, 2, 3\}\), and consider the cycle \(P_{E_1 \intersect E_4} P_{E_2 \intersect E_4} P_{E_3 \intersect E_4}\). From a sampling perspective, this cycle of Gibbs steps yields a sampling algorithm given by Algorithm~\ref{alg: Example of a Gibbs sampler that always updates the fourth coordinate} wherein the value of the fourth coordinate is never used to update a coordinate. Hence, the first three coordinates can be viewed as being updated by the collapsed Gibbs sampler \(P_{E_1}^{(I)}P_{E_2}^{(I)}P_{E_3}^{(I)}\) whose sampling algorithm is given by Algorithm~\ref{alg: Collapsed Gibbs sampler targeting pi_(1,2,3)}.

\begin{algorithm}[H]
	\caption{A sampler that updates the fourth coordinate in every step} 
	\label{alg: Example of a Gibbs sampler that always updates the fourth coordinate}
	\begin{algorithmic}[1]
		\setcounter{ALG@line}{-1}
		\State {\it Input:} Current value \(X^{(n)} = (x_1^{(n)}, x_2^{(n)}, x_3^{(n)}, x_4^{(n)})\).
		\State Draw \((X^{(n+1)}_1, X^{(n+1)}_4) \sim \pi(\cdot, \cdot |x_2^{(n)}, x_{3}^{(n)})\), and call the observed value \((x_{1}^{(n+1)}, x_4^{(n+1)})\).
		\State Draw \((X^{(n+1)}_2, X^{(n+1)}_4) \sim \pi(\cdot, \cdot |x_1^{(n+1)}, x_{3}^{(n)})\), and call the observed value \((x_{2}^{(n+1)}, x_4^{(n+1)})\).
		\State Draw \((X^{(n+1)}_3, X^{(n+1)}_4) \sim \pi(\cdot, \cdot |x_1^{(n+1)}, x_{2}^{(n+1)})\), and call the observed value \((x_{3}^{(n+1)}, x_4^{(n+1)})\).
		\State Set \(n = n + 1\).
	\end{algorithmic}
\end{algorithm}

\begin{algorithm}[H]
	\caption{Collapsed Gibbs sampler targeting \(\pi_{\{1,2,3\}}\)} 
	\label{alg: Collapsed Gibbs sampler targeting pi_(1,2,3)}
	\begin{algorithmic}[1]
		\setcounter{ALG@line}{-1}
		\State {\it Input:} Current value \(X^{(n)} = (x_1^{(n)}, x_2^{(n)}, x_3^{(n)})\).
		
		\State Draw \(X^{(n+1)}_1 \sim \pi(\cdot|x_2^{(n)}, x_{3}^{(n)})\), and call the observed value \(x_{1}^{(n+1)}\).
		
		\State Draw \(X^{(n+1)}_2 \sim \pi(\cdot|x_1^{(n+1)}, x_{3}^{(n)})\), and call the observed value \(x_{2}^{(n+1)}\).
		
		\State Draw \(X^{(n+1)}_3 \sim \pi(\cdot|x_1^{(n+1)}, x_{2}^{(n+1)})\), and call the observed value \(x_{3}^{(n+1)}\).
		
		\State Set \(n = n + 1\).
	\end{algorithmic}
\end{algorithm}

From a spectral perspective, the two samplers are indeed the same as we now show. To every cycle and mixture of collapsed Gibbs steps, there exists a corresponding cycle and mixture of Gibbs steps that is spectrally equivalent in the following sense.

\begin{prop}
	\label{prop: spectral connection between collapsed Gibbs and Gibbs}
	Let \(J_1, \ldots, J_g\) be nonempty proper subsets of \(I\) satisfying \(\bigcup_{d = 1}^g J_d = I\). Let \(F = \bigcap_{i \in [K] \setminus I} E_i\). For \(d \in [g]\), let \(F_d = \bigcap_{j \in J_d} E_j\) and \(F_d^{(I)} = \bigcap_{j \in J_d} E_{j}^{(I)}\). Then,
	\begin{equation}
		\label{eq: spectrum of cycle of collapsed Gibbs steps coincides with spectrum of a cycle of Gibbs steps}
		\sigma\paren*{P_{F_1 \intersect F} \cdots P_{F_g \intersect F} - \Pi} = \sigma\paren*{P_{F_1^{(I)}} \cdots P_{F_g^{(I)}} - \Pi_I},
	\end{equation}
	and, for all \((v_1, \ldots, v_g) \in \mscr{S}_g\),
	\begin{equation}
		\label{eq: spectrum of mixture of collapsed Gibbs steps coincides with spectrum of a mixture of Gibbs steps}
		\sigma\paren*{\sum_{d = 1}^g v_d P_{F_d \intersect F} - \Pi} = \sigma\paren*{\sum_{d = 1}^g v_d P_{F_d^{(I)}} - \Pi_I}.
	\end{equation}
\end{prop}

\begin{rem}
	\label{rem: same rate in TVD}
	Let \((X^{(n)})_{n \geq 0}\), where \(X^{(n)} = (X_I^{(n)}, X_{-I}^{(n)})\) for \(n \geq 0\), denote the Markov chain with operator \(P_{F_1 \intersect F} \cdots P_{F_g \intersect F}\) (resp. \(\sum_{d = 1}^g v_d P_{F_d \intersect F}\)) as given in Proposition~\ref{prop: spectral connection between collapsed Gibbs and Gibbs}. The process \((X_I^{(n)})_{n \geq 0}\) is a Markov chain with operator \(P_{F_1^{(I)}} \cdots P_{F_g^{(I)}}\) (resp. \(\sum_{d = 1}^g v_d P_{F_d}^{(I)}\)). In addition to sharing the same spectrum in the sense of Proposition~\ref{prop: spectral connection between collapsed Gibbs and Gibbs}, these two chains also share the same convergence behavior. By Corollary~2 of \cite{Roberts:Rosenthal:2001}, both chains converge to stationarity at the same rate in total variation distance since the chain \((X_I^{(n)})_{n \geq 0}\) is a ``homogeneously functionally de-initializing chain'' for \((X^{(n)})_{n \geq 0}\).
\end{rem}

To prove Proposition~\ref{prop: spectral connection between collapsed Gibbs and Gibbs}, we use the following result from \cite{Rosenthal:Rosenthal:2015} to relate certain Gibbs steps to their restrictions to \(F = \bigcap_{i \in [K] \setminus I} E_i\). Proposition~\ref{prop: similarity between collapsed Gibbs steps and Gibbs steps with restricted domain} then provides the connection between the restrictions and collapsed Gibbs steps.

\begin{lem}[{\cite{Rosenthal:Rosenthal:2015}, Proposition~5.2}]
	\label{lem: Rosenthal/Rosenthal Proposition 5.2}
	Let \(T\) be a bounded linear operator on a complex Hilbert space \(H\). Suppose \(M\) is a proper closed subspace of \(H\) which contains the range of \(T\). Then,
	\begin{equation*}
		\sigma(T) = \sigma(T|_M) \union \{0\}.
	\end{equation*}
\end{lem}

\begin{rem}
	\label{rem: equivalent definition of spectral gap}
	For a Markov operator \(Q \in B(L^2(\pi))\) satisfying \(\Pi Q = \Pi\), some authors say that \(Q\) has a spectral gap when \(r(Q|_{E^\perp}) < 1\) instead of requiring that \(r(Q - \Pi) < 1\). See, e.g., \cite{Roberts:Rosenthal:1997} and \cite{Qin:Hobert:Khare:2019}. Lemma~\ref{lem: Rosenthal/Rosenthal Proposition 5.2} provides a straightforward way to see that these two definitions are equivalent. Since \(P_{E^\perp}(Q - \Pi) = (\Id - \Pi)(Q - \Pi) = Q - \Pi\), the range of \(Q - \Pi\) is contained in \(E^\perp\), so \(\sigma(Q - \Pi) = \sigma(Q|_{E^\perp}) \union \{0\}\). Hence, \(r(Q - \Pi) = r(Q|_{E^\perp})\). 
\end{rem}

\begin{proof}[Proof of Proposition~\ref{prop: spectral connection between collapsed Gibbs and Gibbs}]
	Let \(\Phi \colon L^2(\pi_I) \to F\) be the isometric isomorphism given in Proposition~\ref{prop: L^2(pi_I) and bigcap_-I E_i are isometrically isomorphic}. From Lemma~\ref{lem: Rosenthal/Rosenthal Proposition 5.2},
	\begin{equation*}
		\sigma(P_{F_1 \intersect F} \cdots P_{F_g \intersect F} - \Pi) = \sigma\paren*{(P_{F_1 \intersect F} \cdots P_{F_g \intersect F} - \Pi)|_F} \union \{0\},
	\end{equation*}
	and, from Proposition~\ref{prop: similarity between collapsed Gibbs steps and Gibbs steps with restricted domain},
	\begin{align*}
		(P_{F_1 \intersect F} \cdots P_{F_g \intersect F} - \Pi)|_F &= (\Phi P_{F_1^{(I)}}\Phi^{-1}) \cdots (\Phi P_{F_g^{(I)}}\Phi^{-1}) - \Phi\Pi_I \Phi^{-1} \\
		&= \Phi(P_{F_1^{(I)}} \cdots P_{F_g^{(I)}} - \Pi_I)\Phi^{-1}.
	\end{align*}
	For any \(\lambda \in \mbb{C}\), the operator \(Q_\lambda \coloneq (P_{F_1 \intersect F} \cdots P_{F_g \intersect F} - \Pi)|_F - \lambda \, \Id_{F}\) is invertible if and only if \((P_{F_1^{(I)}} \cdots P_{F_g^{(I)}} - \Pi_I) - \lambda \, \Id_{L^2(\pi_I)} = \Phi^{-1}Q_\lambda \Phi\) is invertible. Hence,
	\begin{equation*}
		\sigma\paren*{(P_{F_1 \intersect F} \cdots P_{F_g \intersect F} - \Pi)|_F} = \sigma\paren*{P_{F_1^{(I)}} \cdots P_{F_g^{(I)}} - \Pi_I}.
	\end{equation*}
	Since the operator \(P_{F_1^{(I)}} \cdots P_{F_g^{(I)}} - \Pi_I\) maps every constant function in \(L^2(\pi_I)\) to 0, it fails to be injective, so zero belongs in \(\sigma\paren{P_{F_1^{(I)}} \cdots P_{F_g^{(I)}} - \Pi_I}\). This establishes \eqref{eq: spectrum of cycle of collapsed Gibbs steps coincides with spectrum of a cycle of Gibbs steps}. A similar argument obtains \eqref{eq: spectrum of mixture of collapsed Gibbs steps coincides with spectrum of a mixture of Gibbs steps} since 
	\begin{equation*}
		\left.\paren*{\sum_{d = 1}^g v_d P_{F_d \intersect F} - \Pi}\right|_F = \Phi\paren*{\sum_{d = 1}^g v_d P_{F_d^{(I)}} - \Pi_I}\Phi^{-1}
	\end{equation*}
	and since \(\sum_{d = 1}^g v_d P_{F_d^{(I)}} - \Pi_I\) is not injective.
\end{proof}

Proposition~\ref{prop: spectral connection between collapsed Gibbs and Gibbs} implies that a cycle (resp. mixture) of collapsed Gibbs steps has a spectral gap if and only if a corresponding cycle (resp. mixture) of Gibbs steps has a spectral gap, in which case the gaps are equal. Recall from Theorem~\ref{thm: inheritance of a spectral gap from the full Gibbs sampler} that cycles and mixtures of Gibbs steps inherit a spectral gap from a full Gibbs sampler; consequently, this inheritance extends to cycles and mixtures of collapsed Gibbs steps.

\begin{cor}
	\label{cor: inheritance of a spectral gap for cycles and mixtures of collapsed Gibbs steps}
	If a full Gibbs sampler has a spectral gap, then every cycle and mixture of collapsed Gibbs steps has a spectral gap.
\end{cor}

One may opt to analyze a cycle or mixture of collapsed Gibbs steps instead of a cycle or mixture of Gibbs steps if the conditional distributions associated with a marginal distribution are easier to handle. As Proposition~\ref{prop: spectral connection between collapsed Gibbs and Gibbs} shows, the analysis of a collapsed sampler that targets a marginal distribution yields information about a related sampler that targets the entire distribution \(\pi\). While the theoretical usefulness of the latter sampler is demonstrated in establishing Corollary~\ref{cor: inheritance of a spectral gap for cycles and mixtures of collapsed Gibbs steps}, such a sampler is nonstandard in practice due to its often infeasible implementation. In an algorithm like Algorithm~\ref{alg: Example of a Gibbs sampler that always updates the fourth coordinate}, it is seldom possible from a computational standpoint that a certain coordinate can be updated jointly in every step of the algorithm. In view of this issue, we identify a setting in Section~\ref{subsection: Commuting pairs of Gibbs steps} where a collapsed deterministic-scan Gibbs sampler is intimately related with a blocked deterministic-scan Gibbs sampler.

\subsection{Commuting pairs of Gibbs steps}
\label{subsection: Commuting pairs of Gibbs steps}

For this subsection, view \(\mcal{X}\) as any fixed three-fold product \(\mcal{U} \cross \mcal{V} \cross \mcal{W}\). (To be precise, take \(\mcal{U} = \mcal{X}_{I_1}\), take \(\mcal{V} = \mcal{X}_{I_2}\), and take \(\mcal{W} = \mcal{X}_{I_3}\), where \(\{I_1, I_2, I_3\}\) is some fixed partition of \([K]\), and then work exclusively on \(\mcal{U} \cross \mcal{V} \cross \mcal{W}\).) Let \((U,V,W)\) be a random element taking values in \(\mcal{U} \cross \mcal{V} \cross \mcal{W}\) with distribution \(\pi\). Define \(E_U\), \(E_V\), and \(E_W\) to be the closed subspaces of functions in \(L^2(\pi)\) that are constant with respect to the \(u\)th, \(v\)th, and \(w\)th coordinate, respectively. 

Recall that two orthogonal projections with ranges \(M_1\) and \(M_2\), say, commute if and only if  \(P_{M_1}P_{M_2} = P_{M_1 \intersect M_2}\). (See, e.g., \cite{Halmos:1951}, Section~29.) For a pair of Gibbs steps, commutativity has the following probabilistic interpretation.

\begin{prop}
	\label{prop: U cond. indep. of V given W iff P_(F_U) and P_(F_V) commute}
	The Gibbs steps \(P_{E_U}\) and \(P_{E_V}\) commute if and only if \(U\) and \(V\) are conditionally independent given \(W\).
\end{prop}

\begin{proof}
	First, suppose that \(U\) and \(V\) are conditionally independent given \(W\). For all \(f \in L^2(\pi)\) and all \((u,v,w) \in \mcal{U} \cross \mcal{V} \cross \mcal{W}\), conditional independence of \(U\) and \(V\) given \(W\) implies that
	\begin{align*}
		[P_{E_U}P_{E_V}f](u,v,w) &= \int_{\mcal{U}} \bracket*{\int_{\mcal{V}} f(u', v', w) \, \pi(dv'|u', w)} \, \pi(du'|v, w) \\
		&= \int_{\mcal{U}} \bracket*{\int_{\mcal{V}} f(u', v', w) \, \pi(dv'|w)} \, \pi(du'|w) \\
		&= \int_{\mcal{U} \cross \mcal{V}} f(u', v', w) \pi(du', dv'|w) \\
		&= [P_{E_U \intersect E_V}f](u,v,w).
	\end{align*}
	Hence, \(P_{E_U}P_{E_V} = P_{E_U \intersect E_V}\), so \(P_{E_U}\) commutes with \(P_{E_V}\).
	
	Conversely, suppose that \(P_{E_U}\) and \(P_{E_V}\) commute, so \(P_{E_U}P_{E_V} = P_{E_U \intersect E_V}\). Fix any \(A_U \in \mscr{B}(\mcal{U})\) and any \((u,v,w) \in \mcal{U} \cross \mcal{V} \cross \mcal{W}\). Observe that
	\begin{align*}
		&[P_{E_U}P_{E_V}\mathbf{1}_{A_U \cross \mcal{V} \cross \mcal{W}}](u,v,w) \\
		&\qquad = \int_{\mcal{U}} \bracket*{\int_{\mcal{V}} \mathbf{1}_{A_U \cross \mcal{V} \cross \mcal{W}}(u', v', w) \, \pi(dv'|u', w)} \, \pi(du'|v, w) \\
		&\qquad = \int_{A_U} \pi(du'|v, w) \\
		&\qquad = \Pr(U \in A_U|V = v, W = w)
	\end{align*}
	and
	\begin{align*}
		[P_{E_U \intersect E_V}\mathbf{1}_{A_U \cross \mcal{V} \cross \mcal{W}}](u,v,w) &= \int_{\mcal{U} \cross \mcal{V}} \mathbf{1}_{A_U \cross \mcal{V} \cross \mcal{W}}(u', v', w) \, \pi(du', dv'|w) \\
		&= \int_{A_U} \pi(du'|w) \\
		&= \Pr(U \in A_U|W = w).
	\end{align*}
	Since \(P_{E_U}P_{E_V} = P_{E_U \intersect E_V}\), it follows that
	\begin{equation*}
		\Pr(U \in A_U|V = v, W = w) = \Pr(U \in A_U|W = w)
	\end{equation*}
	for all \(A_U \in \mscr{B}(\mcal{U})\) and all \((v, w) \in \mcal{V} \cross \mcal{W}\). This completes the proof.
\end{proof}

Let \(\pi_{UW}\) and \(\pi_{VW}\) denote the marginal distributions of \(\pi\) on \(\mcal{U} \cross \mcal{W}\) and \(\mcal{V} \cross \mcal{W}\), respectively. Further, let \(\Pi_{UW} \in B(L^2(\pi_{UW}))\) and \(\Pi_{VW} \in B(L^2(\pi_{VW}))\) be given by \(\Pi_{UW}f_{UW} \equiv \pi_{UW} f_{UW}\) for \(f_{UW} \in L^2(\pi_{UW})\) and \(\Pi_{VW}f_{VW} \equiv \pi_{VW}f_{VW}\) for \(f_{VW} \in L^2(\pi_{VW})\). Define \(E_U^{(UW)}\) and \(E_W^{(UW)}\) to be the closed subspaces of functions in \(L^2(\pi_{UW})\) that are constant with respect to the \(u\)th and \(w\)th coordinate, respectively, and analogously define the closed subspaces \(E_V^{(VW)}\) and \(E_W^{(VW)}\) of \(L^2(\pi_{VW})\).

\begin{cor}
	\label{cor: spectral connection between collapsed gibbs sampler and blocked gibbs sampler when U and V are cond. indep. given W}
	If \(U\) and \(V\) are conditionally independent given \(W\), then
	\begin{align*}
		\sigma\paren*{P_{E_U} P_{E_V \intersect E_W} - \Pi} &= \sigma\paren*{P_{E_U}^{(UW)} P_{E_W}^{(UW)} - \Pi_{UW}} \intertext{and} 
		\sigma\paren*{P_{E_V} P_{E_U \intersect E_W} - \Pi} &= \sigma\paren*{P_{E_V}^{(VW)} P_{E_W}^{(VW)} - \Pi_{VW}}.
	\end{align*}
\end{cor}

\begin{proof}
	By Proposition~\ref{prop: U cond. indep. of V given W iff P_(F_U) and P_(F_V) commute}, the operators \(P_{E_U}\) and \(P_{E_V}\) commute, so \(P_{E_U} P_{E_V} = P_{E_U \intersect E_V}\). Hence, 
	\begin{equation*}
		P_{E_U}P_{E_V \intersect E_W} = P_{E_U} P_{E_V} P_{E_V \intersect E_W} = P_{E_U \intersect E_V} P_{E_V \intersect E_W},
	\end{equation*}
	and, similarly, \(P_{E_V} P_{E_U \intersect E_W} = P_{E_V \intersect E_U} P_{E_U \intersect E_W}\). The claim now follows from Proposition~\ref{prop: spectral connection between collapsed Gibbs and Gibbs}.
\end{proof}

Suppose that \(U\) and \(V\) are conditionally independent given \(W\). Recall that the algorithm of the blocked Gibbs sampler with operator \(P_{E_U}P_{E_V \intersect E_W}\) updates a current iteration \((U_n, V_n, W_n) = (u_n, v_n, w_n)\) by first updating \(u_n\) with a draw \(u_{n+1}\) from \(\pi(\cdot|w_n)\) and then updating \((v_n, w_n)\) with a draw \((v_{n+1}, w_{n+1})\) from \(\pi(\cdot, \cdot|u_{n+1})\). A practical consequence of Corollary~\ref{cor: spectral connection between collapsed gibbs sampler and blocked gibbs sampler when U and V are cond. indep. given W} is that no information is lost by ignoring the \(v\)th coordinate when determining whether the blocked Gibbs sampler has a spectral gap. It suffices to focus one's attention on the collapsed sampler that updates a current iteration \((U_n, W_n) = (u_n, w_n)\) by first updating \(u_n\) with a draw \(u_{n+1}\) from \(\pi(\cdot|w_n)\) and then updating \(w_n\) with a draw \(w_{n+1}\) from \(\pi(\cdot|u_{n+1})\). Analogously, one can safely disregard the \(u\)th coordinate when analyzing the blocked Gibbs sampler with operator \(P_{E_V}P_{E_U \intersect E_W}\). We provide an example in Section~\ref{subsection: Example} that illustrates (among other matters) how restricting attention to the collapsed chain can simplify analysis.

\subsection{Marginal chains of two-component deterministic-scan Gibbs samplers}
\label{subsection: Marginal chains of two-component Gibbs samplers}

It is a well-known result from \cite{Liu:Wong:Kong:1994} that a two-component deterministic-scan Gibbs sampler has a spectral gap if and only if its marginal chain has a spectral gap, in which case the gaps are equal. Through different means, \cite{Rosenthal:Rosenthal:2015} showed that the spectrum of a two-component deterministic-scan Gibbs sampler coincides with the spectrum of the operator of its marginal chain with a possible exception at zero. However, this connection regarding the spectrum of a Gibbs sampler does not immediately imply the aforementioned result from \cite{Liu:Wong:Kong:1994}: a spectral gap concerns the spectrum of a Markov operator minus \(\Pi\), not just the spectrum of a Markov operator. To fill in this detail in the operator-theoretic manner of \cite{Rosenthal:Rosenthal:2015}, we establish a spectral description akin to Proposition~\ref{prop: spectral connection between collapsed Gibbs and Gibbs} that proves in a unified fashion that the spectral gaps for a two-component deterministic-scan Gibbs sampler, its adjoint, and their marginal chains all agree whenever they exist. 

View \(\mcal{X}\) as any fixed two-fold product \(\mcal{Y} \cross \mcal{Z}\). (To be precise, take \(\mcal{Y} = \mcal{X}_I\) and \(\mcal{Z} = \mcal{X}_{-I}\), where \(I\) is some fixed nonempty proper subset of \([K]\).) Define \(E_Y\) and \(E_Z\) to be the closed subspaces of functions in \(L^2(\pi)\) that are constant with respect to the \(y\)th and \(z\)th coordinate, respectively. Further, let \(\pi_Y\) and \(\pi_Z\) denote the marginal distributions of \(\pi\) on \(\mcal{Y}\) and \(\mcal{Z}\), respectively. Define \(\Pi_{Y} \in B(L^2(\pi_Y))\) and \(\Pi_Z \in B(L^2(\pi_Z))\) by \(\Pi_Y f_Y \equiv \pi_Y f_Y\) for \(f_Y \in L^2(\pi_Y)\) and \(\Pi_Z f_Z \equiv \pi_Z f_Z\) for \(f_Z \in L^2(\pi_Z)\).

The two-component deterministic-scan Gibbs sampler with operator \(P_{E_Y}P_{E_Z}\) defines a kernel that only depends on the \(z\)th coordinate of the current iteration. Hence, there is a corresponding marginal Markov chain on \(\mcal{Z}\) with transition kernel given by
\begin{equation}
	\label{eq: kernel of the marginal chain (Q_Z)}
	Q_Z(z, A_Z) = (P_{E_Y}P_{E_Z}\mathbf{1}_{\mcal{Y} \cross A_Z})(y^*, z) = \int_{\mcal{Y}} \bracket*{ \int_{A_Z}  \pi(dz'|y')} \pi(dy'|z) 
\end{equation}
for all \((z, A_Z) \in \mcal{Z} \cross \mscr{B}(\mcal{Z})\), where \(y^* \in \mcal{Y}\) is an inconsequential fixed element. The kernel given by \eqref{eq: kernel of the marginal chain (Q_Z)} is reversible with respect to \(\pi_Z\) (\cite{Liu:Wong:Kong:1994}, Lemma~3.1); consequently, its corresponding Markov operator \(Q_Z \in B(L^2(\pi_Z))\) is self-adjoint. There is an analogous marginal chain on \(\mcal{Y}\) associated with the Gibbs sampler with operator \(P_{E_Z}P_{E_Y}\). Its transition kernel is given by
\begin{equation}
	\label{eq: kernel of the marginal chain (Q_Y)}
	Q_Y(y, A_Y) = (P_{E_Z}P_{E_Y}\mathbf{1}_{A_Y \cross \mcal{Z}})(y, z^*) = \int_{\mcal{Z}} \bracket*{\int_{A_Y} \pi(dy'|z')} \pi(dz'|y)
\end{equation}
for all \((y, A_Y) \in \mcal{Y} \cross \mscr{B}(\mcal{Y})\), where \(z^* \in \mcal{Z}\) is an inconsequential fixed element. As with \(Q_Z\), the corresponding Markov operator \(Q_Y \in B(L^2(\pi_Y))\) is self-adjoint since the kernel given by \eqref{eq: kernel of the marginal chain (Q_Y)} is reversible with respect to \(\pi_Y\).

The spectral relationship between two-component deterministic-scan Gibbs samplers and their marginal chains is given in the following result.

\begin{prop}
	\label{prop: spectral connection between two-component deterministic-scan Gibbs and its marginal chain}
	Let \(Q_Y \in B(L^2(\pi_Y))\) and \(Q_Z \in B(L^2(\pi_Z))\) be the corresponding Markov operators of the kernels given by \eqref{eq: kernel of the marginal chain (Q_Y)} and \eqref{eq: kernel of the marginal chain (Q_Z)}, respectively. Then,
	\begin{equation*}
		\sigma(P_{E_Y}P_{E_Z} - \Pi) = \sigma(P_{E_Z}P_{E_Y} - \Pi) = \sigma(Q_Y - \Pi_Y) = \sigma(Q_Z - \Pi_Z).
	\end{equation*}	
\end{prop}

\begin{proof}
	We first show that
	\begin{equation}
		\label{eq: spectra of P_(E_Z)P_(E_Y) - Pi and sigma(Q_Y) - Pi_Y are equal}
		\sigma(P_{E_Z}P_{E_Y} - \Pi) = \sigma(Q_Y - \Pi_Y).
	\end{equation}
	The argument is similar to the one given in the proof of Proposition~\ref{prop: spectral connection between collapsed Gibbs and Gibbs}. Letting \(\Phi \colon L^2(\pi_Y) \to E_Z\) be the isometric isomorphism given in Proposition~\ref{prop: L^2(pi_I) and bigcap_-I E_i are isometrically isomorphic}, observe that 
	\begin{align*}
		(P_{E_Z}P_{E_Y} - \Pi)|_{E_Z} = \Phi(Q_Y - \Pi_Y)\Phi^{-1}.
	\end{align*}
	Hence, by Lemma~\ref{lem: Rosenthal/Rosenthal Proposition 5.2} and by similarity of the operators,
	\begin{equation*}
		\sigma(P_{E_Z}P_{E_Y} - \Pi) = \sigma((P_{E_Z}P_{E_Y} - \Pi)|_{E_Z}) \union \{0\} = \sigma(Q_Y - \Pi_Y) \union \{0\}.
	\end{equation*}
	Noting that \(Q_Y - \Pi_Y\) is not injective yields \eqref{eq: spectra of P_(E_Z)P_(E_Y) - Pi and sigma(Q_Y) - Pi_Y are equal}. By a symmetrical argument,
	\begin{equation*}
		\sigma(P_{E_Y}P_{E_Z} - \Pi) = \sigma(Q_Z - \Pi_Z).
	\end{equation*}
	
	Now, for any bounded linear operator \(T\) on a complex Hilbert space, recall that \(\sigma(T^*) = \{\overline{\lambda} : \lambda \in T\}\) and that \(\sigma(T) \subset \mbb{R}\) whenever \(T\) is self-adjoint. Hence, the spectra of \(Q_Y - \Pi_Y\) and \(Q_Z - \Pi_Z\) are contained in \(\mbb{R}\). Moreover, these facts in conjunction with \eqref{eq: spectra of P_(E_Z)P_(E_Y) - Pi and sigma(Q_Y) - Pi_Y are equal} imply that
	\begin{align*}
		\sigma(P_{E_Y}P_{E_Z} - \Pi) &= \{\overline{\lambda} : \lambda \in \sigma(P_{E_Z}P_{E_Y} - \Pi)\} \\
		&= \{\overline{\lambda} : \lambda \in \sigma(Q_Y - \Pi_Y)\} \\
		&= \sigma(Q_Y - \Pi_Y);
	\end{align*}
	it similarly follows that \(\sigma(P_{E_Z}P_{E_Y} - \Pi) = \sigma(Q_Z - \Pi_Z)\). This completes the proof.
\end{proof}

\begin{rem}
	\label{rem: same L2 convergence rates when U, V cond. independent given W}
	Consider once more the setting in Corollary~\ref{cor: spectral connection between collapsed gibbs sampler and blocked gibbs sampler when U and V are cond. indep. given W} wherein the random elements \(U\) and \(V\) are conditionally independent given \(W\). Recall that \(P_{E_U}P_{E_V \intersect E_W}\) has a spectral gap if and only if \(P_{E_U}^{(UW)}P_{E_W}^{(UW)}\) has a spectral gap, in which case the gaps of these two deterministic-scan Gibbs samplers are equal.  \cite{Qin:Jones:2022} showed that the \(L^2\)-convergence rate of a two-component deterministic-scan Gibbs sampler is equal to the spectral gap of its marginal chain. Hence, in addition to having the same rate of convergence to stationarity in total variation distance (as noted in Remark~\ref{rem: same rate in TVD}), the Markov chain \(((U^{(n)}, V^{(n)}, W^{(n)}))_{n \geq 0}\) with operator \(P_{E_U}P_{E_V \intersect E_W}\) and the Markov chain \(((U^{(n)}, W^{(n)}))_{n \geq 0}\) with operator \(P_{E_U}^{(UW)}P_{E_W}^{(UW)}\) have the same \(L^2\)-convergence rate. Analogously, the chain \(((U^{(n)}, V^{(n)}, W^{(n)}))_{n \geq 0}\) with operator \(P_{E_V}P_{E_U \intersect E_W}\) and the chain \(((V^{(n)}, W^{(n)}))_{n \geq 0}\) with operator \(P_{E_V}^{(VW)}P_{E_W}^{(VW)}\) share the same \(L^2\)-convergence rate as well.
\end{rem}

Since every two-component deterministic-scan Gibbs sampler inherits a spectral gap from a full Gibbs sampler (Theorem~\ref{thm: introduction, inheritence of a spectral gap}), the following result directly follows from Proposition~\ref{prop: spectral connection between two-component deterministic-scan Gibbs and its marginal chain}. It is immaterial whether the two-component deterministic-scan Gibbs sampler is formed by Gibbs steps or collapsed Gibbs steps.

\begin{cor}
	If a full Gibbs sampler has a spectral gap, then every marginal chain that corresponds to some two-component deterministic-scan Gibbs sampler, formed either by Gibbs steps or collapsed Gibbs steps, has a spectral gap.
\end{cor}

Contrary to what is stated in the introduction of \cite{Chlebicka:Latuszynski:Miasojedow:2025}, the Markov kernels of two-component deterministic-scan Gibbs samplers are not necessarily reversible. Indeed, by a similar argument to the one given in the proof of Proposition~\ref{prop: U cond. indep. of V given W iff P_(F_U) and P_(F_V) commute}, the operator \(P_{E_Y} P_{E_Z}\) is self-adjoint if and only if \(Y\) and \(Z\) are independent, where \((Y,Z)\) is a random element taking values in \(\mcal{Y} \cross \mcal{Z}\) with distribution \(\pi\). Whenever \(Y\) and \(Z\) are not independent, which is the case most of the time, \(P_{E_Y} P_{E_Z}\) and \(P_{E_Z}P_{E_Y}\) define kernels that are not reversible with respect to \(\pi\). The redeeming quality, of course, is that two-component deterministic-scan Gibbs samplers share many properties with their marginal chains that are always reversible.

\section{A spectral disconnection}
\label{section: A spectral disconnection}

Despite blocked Gibbs samplers inheriting a spectral gap from a full Gibbs sampler (Theorem~\ref{thm: inheritance of a spectral gap from the full Gibbs sampler}), an example is given in this section to show that such a connection need not hold between two different blocked Gibbs samplers: even if the sizes of the blocks are the same, one blocked Gibbs sampler may have a spectral gap while another blocked Gibbs sampler does not. Our example also shows that a collapsed Gibbs sampler may fail to inherit a spectral gap from another collapsed Gibbs sampler. In fact, we show that one blocked (resp. collapsed) Gibbs sampler can be geometrically ergodic while another blocked (resp. collapsed) Gibbs sampler is not.

\subsection{Background on geometric ergodicity}
\label{subsection: Background on geometric ergodicity}

Recall (from Section~\ref{section: A generalized solidarity principle of the spectral gap}) that the existence of a spectral gap for a Markov operator ensures that the sequence generated by the powers of the operator converges in norm to \(\Pi\). On the probabilistic front, the existence of a spectral gap for a Markov operator also ensures that the sequence of the corresponding \(n\)-step transition kernels converges exponentially fast in total variation distance to \(\pi\). To phrase this more precisely, some definitions are reviewed. 

For any Markov kernel \(Q(x,dy)\) defined on \(\mcal{X} \cross \mscr{B}(\mcal{X})\), denote the \(n\)-step transition kernel by \(Q^n(x, dy)\) (with \(Q^1(x,dy) = Q(x,dy)\)). The kernel \(Q(x,dy)\) is \emph{\(\varphi\)-irreducible} if there exists some nonzero \(\sigma\)-finite measure \(\varphi\) on \(\mscr{B}(\mcal{X})\) with the following property: for all \(x \in \mcal{X}\) and all \(A \in \mscr{B}(\mcal{X})\) with \(\varphi(A) > 0\), there exists an integer \(n \geq 1\) satisfying \(Q^n(x, A) > 0\). Any \(\varphi\)-irreducible kernel \(Q(x,dy)\) admits a \(\sigma\)-finite measure \(\psi\) on \(\mcal{X}\), called the \emph{maximal irreducibility measure}, such that \(Q(x,dy)\) is \(\psi\)-irreducible. The maximal irreducibility measure \(\psi\) is unique in the sense that any \(\sigma\)-finite measure \(\varphi'\) such that \(Q(x,dy)\) is \(\varphi'\)-irreducible must be absolutely continuous with respect to \(\psi\) (\cite{Meyn:Tweedie:2009}, Proposition~4.2.2). A \(\psi\)-irreducible kernel \(Q(x,dy)\) is \emph{periodic} if there exists a collection of nonempty mutually disjoint sets \(A_0, \ldots, A_{m-1} \in \mscr{B}(\mcal{X})\) with \(m > 1\) such that \(\psi((\bigcup_{j = 0}^{m-1} A_j)^c) = 0\) and \(Q(x, A_{j+1}) = 1\) for all \(x \in A_j\) and all \(j = 0, \ldots, m - 1 \, \pmod m\). A \(\psi\)-irreducible kernel \(Q(x,dy)\) that is not periodic is called \emph{aperiodic}.

The following results are collected in Theorem~1 of the rejoinder in \cite{Tierney:1994}. Let \(Q(x,dy)\) be a Markov kernel with stationary distribution \(\pi\). If \(Q(x,dy)\) is \(\varphi\)-irreducible, then \(Q(x,dy)\) is \(\pi\)-irreducible, \(\varphi\) is absolutely continuous with respect to \(\pi\), and \(\pi\) is the unique invariant probability measure of \(Q(x,dy)\). If \(Q(x,dy)\) is also aperiodic, then     
\begin{equation*}
	\label{eq: lim_n ||Q^n(x, .) - pi(.)||_TV = 0}
	\lim_{n \to \infty} \norm{Q^n(x,\cdot) - \pi(\cdot)}_{\mathrm{TV}} = 0 
\end{equation*}
for \(\pi\)-a.e. \(x \in \mcal{X}\), where \(\norm{Q^n(x,\cdot) - \pi(\cdot)}_{\mathrm{TV}} \coloneq \sup_{A \in \mscr{B}(\mcal{X})}\abs{Q^n(x,A) - \pi(A)}\) is the total variation distance between the probability measures \(Q^n(x,\cdot)\) and \(\pi\).  

In view of the aforementioned facts from \cite{Tierney:1994}, the maximal irreducibility measure is \(\pi\) (up to equivalence of measures) for any \(\varphi\)-irreducible Markov kernel with invariant probability measure \(\pi\). This property is used in the proof of the following result to show that \(\varphi\)-irreducibility ensures aperiodicity for any Markov kernel corresponding to a cycle or mixture of Gibbs steps. The argument is essentially that used to prove Theorem~4.6 of \cite{Geyer:2005} and is provided in Appendix~\ref{section: proof of aperiodic proposition} for completeness.

\begin{prop}
	\label{prop: Q(x,dy) is aperiodic if Q is pi-irreducible}
	Let \(Q\) be a cycle or mixture of Gibbs steps. If \(Q(x,dy)\) is \(\varphi\)-irreducible, then \(Q(x,dy)\) is aperiodic.
\end{prop}

A \(\varphi\)-irreducible and aperiodic kernel \(Q(x,dy)\) with invariant probability measure \(\pi\) is \emph{geometrically ergodic} if there exist a constant \(\rho \in [0,1)\) and a function \(C : \mcal{X} \to (0, \infty)\) such that 
\begin{equation*}
	\label{eq: geometric ergodicity}
	\norm{Q^n(x, \cdot) - \pi(\cdot)}_{\mathrm{TV}} \leq C(x) \rho^n, \quad n \geq 1,
\end{equation*}
for \(\pi\)-a.e. \(x \in \mcal{X}\). Recall that for any \(\varphi\)-irreducible and aperiodic kernel \(Q(x,dy)\) with invariant probability measure \(\pi\), the presence of a spectral gap for the corresponding Markov operator \(Q \in B(L^2(\pi))\) implies geometric ergodicity of \(Q(x,dy)\) (\cite{Kontoyiannis:Meyn:2012}). When \(Q(x,dy)\) is additionally reversible or, equivalently, when \(Q\) is additionally self-adjoint, then the reverse implication also holds: geometric ergodicity of \(Q(x,dy)\) implies the presence of a spectral gap for \(Q\) (\cite{Roberts:Rosenthal:1997}, \cite{Roberts:Tweedie:2001}).

\subsection{Example}
\label{subsection: Example}

Here is a setting in which blocking or collapsing a Gibbs sampler in different ways affects whether or not the resulting operator has a spectral gap. For any random variable \(R\), let \(\mcal{L}(R)\) denote the distribution of \(R\). Consider the following Bayesian linear hierarchical model:
\begin{align*}
	Y &= W + \frac{N_1}{\sqrt{U}}, \\
	W &= V + N_2,
\end{align*}
where \(\mcal{L}(N_1) = \mcal{L}(N_2) = N(0,1)\), where \(\mcal{L}(U) = \Gamma(1/2,1/2)\), and where a flat prior is placed on \(V\). (We say \(\mcal{L}(U) = \Gamma(\alpha, \beta)\) for some \(\alpha, \beta > 0\) if the density of \(U\) is proportional to \(u^{\alpha - 1}\exp[-\beta u]\) for \(u > 0\).) The variables \(N_1\), \(N_2\), \(U\), and \(V\) are assumed mutually independent. In particular, this implies that \(U\) and \(W\) are also independent. 

Fix \(y \in \mbb{R}\) throughout this section. The distribution of interest is \(\pi = \mcal{L}((U,V,W)|Y = y)\). Slightly abusing notation, denote the posterior density of \((U,V,W)|Y = y\) by \(\pi(u,v,w|y)\), and, momentarily letting \(f(\cdot)\) represent a generic density, observe that
\begin{equation}
	\label{eq: pi(u,v,w|y)}
	\begin{aligned}[b]
		\pi(u,v,w|y) &\propto f(y|u,w)f(u)f(w|v)f(v) \\
		&\propto  \paren*{\sqrt{u}\exp\bracket*{-\frac{u}{2}(y - w)^2}} \paren*{\frac{1}{\sqrt{u}}\exp\bracket*{-\frac{1}{2}u}} \paren*{\exp\bracket*{-\frac{1}{2}(w - v)^2}}  \\
		&= \paren*{\exp\bracket*{-\frac{1}{2}\paren*{1 + (y-w)^2}u}} \paren*{\exp\bracket*{-\frac{1}{2}(w - v)^2}}
	\end{aligned}
\end{equation}
for all \(u \in \mbb{R}_+ \coloneq (0, \infty)\) and \(v,w \in \mbb{R}\). It is easy to see that 
\begin{equation*}
	\int_{\mbb{R}} \int_{\mbb{R}_+} \int_{\mbb{R}} \paren*{\exp\bracket*{-\frac{1}{2}\paren*{1 + (y-w)^2}u}} \paren*{\exp\bracket*{-\frac{1}{2}(w - v)^2}} \, dv \, du \, dw < \infty,
\end{equation*} 
so \(\pi(u,v,w|y)\) is indeed a proper density. Since \(\pi(u,v,w|y)\) is strictly positive everywhere, any cycle or mixture of Gibbs steps has a kernel that is \(\pi\)-irreducible. Irreducibility similarly holds for cycles and mixtures of collapsed Gibbs steps. (See, e.g., \cite{Tan:Hobert:2009}, Lemma~1.)

Let \(\mcal{U} = \mbb{R}_+\) and let \(\mcal{V} = \mcal{W} = \mbb{R}\). As in Section~\ref{subsection: Commuting pairs of Gibbs steps}, let \((U,V,W)\) be a random element taking values in \(\mcal{U} \cross \mcal{V} \cross \mcal{W}\) with distribution \(\pi\), and let \(E_U\), \(E_V\), and \(E_W\) be the closed subspaces of functions in \(L^2(\pi)\) that are constant with respect to the \(u\)th, \(v\)th, and \(w\)th coordinate, respectively. We investigate whether the two blocked Gibbs samplers with operators \(P_{E_U} P_{E_V \intersect E_W}\) and \(P_{E_V} P_{E_U \intersect E_W}\) have a spectral gap. From \eqref{eq: pi(u,v,w|y)}, observe that \(U\) and \(V\) are conditionally independent given \(W\). Hence, by Corollary~\ref{cor: spectral connection between collapsed gibbs sampler and blocked gibbs sampler when U and V are cond. indep. given W}, it suffices to analyze the collapsed Gibbs samplers with operators \(P_{E_U}^{(UW)}P_{E_W}^{(UW)}\) and \(P_{E_V}^{(VW)}P_{E_W}^{(VW)}\). 

In view of \eqref{eq: pi(u,v,w|y)}, 
\begin{equation*}
	\pi(v,w|y) \propto \paren*{\frac{1}{1+(y-w)^2}} \paren*{\exp\bracket*{-\frac{1}{2}(w-v)^2}}, \quad (v,w) \in \mcal{V} \cross \mcal{W}.
\end{equation*}
Thus, \(\pi(v,w|y)\) is the posterior density of \((V,W)|Y = y\) given by the following hierarchical model:
\begin{align*}
	Y &= W + Z_1 \\
	W &= V + Z_2,
\end{align*}
where \(\mcal{L}(Z_1) = \mathrm{Cauchy}(0,1)\), where \(\mcal{L}(Z_2) = N(0,1)\), and where a flat prior is placed on \(V\). The variables \(V\), \(Z_1\), and \(Z_2\) are assumed mutually independent. This hierarchical model was studied in \cite{Papaspiliopoulos:Roberts:2008} and Theorem~3.6 therein established that the collapsed Gibbs sampler with operator \(P_{E_V}^{(VW)}P_{E_W}^{(VW)}\) is not geometrically ergodic. It is well-known that two-component Gibbs samplers and their marginal chains converge to stationarity at the same rate in total variation distance (\cite{Roberts:Rosenthal:2001}, \cite{Diaconis:Khare:Saloff-Coste:2008}), so the marginal chain is consequently not geometrically ergodic as well. Since the marginal chain is reversible, it follows that the Markov operator of the marginal chain does not have a spectral gap (\cite{Roberts:Tweedie:2001}). Hence, \(P_{E_V}^{(VW)}P_{E_W}^{(VW)}\) does not have a spectral gap by Proposition~\ref{prop: spectral connection between two-component deterministic-scan Gibbs and its marginal chain}, so \(P_{E_V}P_{E_U \intersect E_W}\) does not have a spectral gap either.

In contrast to the blocked Gibbs sampler with operator \(P_{E_V} P_{E_U \intersect E_W}\), the following result establishes that blocking differently results in an operator with a spectral gap.

\begin{prop}
	\label{prop: P_(E_U) P_(E_V intersect E_W) has a spectral gap}
	The blocked Gibbs sampler with operator \(P_{E_U}P_{E_V \intersect E_W}\) has a spectral gap.
\end{prop}

\begin{proof}
	It suffices to show that the marginal chain of the collapsed Gibbs sampler \(P_{E_U}^{(UW)}P_{E_W}^{(UW)}\) is geometrically ergodic. From \eqref{eq: pi(u,v,w|y)}, 
	\begin{equation*}
		\pi(u,w|y) \propto \exp\bracket*{-\frac{1}{2}(1+(y-w)^2)u},
	\end{equation*}
	so \(\mcal{L}(U|W=w, Y=y) = \Gamma(1, (1+(y-w)^2)/2)\) and \(\mcal{L}(W|U=u, Y=y) = N(y, 1/u)\). Letting \(\tau_w = \sqrt{(1+(y-w)^2)/2}\), observe that the marginal chain of the Gibbs sampler with operator \(P_{E_U}^{(UW)}P_{E_W}^{(UW)}\) has transition density given by
	\begin{align*}
		k(w,w')&= \int_0^\infty \pi(w'|u,y) \pi(u|w,y) \, du \\
		&= \int_0^\infty \paren*{\sqrt{\frac{u}{2\pi}} \exp\bracket*{-\frac{u}{2}(w'-y)^2}} \paren*{\tau_w^2 \exp\bracket*{-\frac{1}{2}(1+(y-w)^2)u}}  \, du \\
		&= \frac{\tau_w^2}{\sqrt{2\pi}} \int_0^\infty \sqrt{u} \exp\bracket*{-\frac{1}{2}(1+(y-w)^2+(y-w')^2)u} \, du \\
		&= \frac{\tau_w^2}{\sqrt{2\pi}} \paren*{\frac{2}{1+(y-w)^2+(y-w')^2}}^{3/2} \frac{\sqrt{\pi}}{2} \\
		&= \frac{\tau_w^2}{(1+(y-w)^2+(y-w')^2)^{3/2}}.
	\end{align*}
	Since
	\begin{align*}
		&(1+(y-w)^2+(y-w')^2)^{-3/2} \\
		&\qquad = 2^{-3/2} \paren*{\frac{1 + (y-w)^2}{2} + \frac{(y-w')^2}{2}}^{-3/2} \\
		&\qquad = 2^{-3/2} \bracket*{\paren*{\frac{1 + (y-w)^2}{2}} \paren*{1 + \frac{(y-w')^2}{2} \paren*{\frac{2}{1+(y-w)^2}}}}^{-3/2} \\
		&\qquad = 2^{-3/2} \tau_w^{-3} \bracket*{1 + \frac{(y-w')^2}{2\tau_w^2}}^{-3/2}
	\end{align*}
	for all \(w, w' \in \mbb{R}\), it follows that
	\begin{equation}
		\label{eq: k(w,w')}
		k(w,w') = \tau_w^2 \paren*{2^{-3/2} \tau_w^{-3} \bracket*{1 + \frac{(y-w')^2}{2\tau_w^2}}^{-3/2}} = \frac{1}{\sqrt{8}\tau_w}\bracket*{1 + \frac{(y-w')^2}{2\tau_w^2}}^{-3/2}.
	\end{equation}
	Recall that a location-scale Student's \(t\) distribution with \(2\) degrees of freedom, location parameter \(\mu \in \mbb{R}\), and scale parameter \(\delta > 0\) has density given by
	\begin{equation}
		\label{eq: location-scale t_2 density}
		g(w') = \frac{1}{\sqrt{8} \delta} \bracket*{1 + \frac{(\mu - w')^2}{2\delta^2}}^{-3/2}, \quad w' \in \mbb{R}.
	\end{equation}
	Hence, for fixed \(w \in \mbb{R}\), the function \(w' \mapsto k(w,w')\) given by \eqref{eq: k(w,w')} is precisely the density of a location-scale \(t\) distribution with \(2\) degrees of freedom, location parameter \(y\), and scale parameter \(\tau_w = \sqrt{(1+(y-w)^2)/2}\). Put equivalently, if \(W_0, W_1, \ldots\) denotes the Markov chain with transition density \(k(w,w')\), then 
	\begin{equation}
		\label{eq: L(W_(n+1)|W_n)}
		\mcal{L}(W_{n+1}|W_n = w_n) = \mcal{L}\paren*{y + \sqrt{\frac{1 + (y-w_n)^2}{2}} \, t_2}, \quad n \geq 0,
	\end{equation}
	where \(\mcal{L}(t_2)\) is the Student's \(t\) distribution with 2 degrees of freedom.
	
	We now establish a standard set of drift and minorization conditions (\cite{Rosenthal:1995}, Theorem~12) from which it follows that the chain \(W_0, W_1, \ldots\) is geometrically ergodic. Define \(V(w) = \sqrt{\abs{y-w}}\) for \(w \in \mbb{R}\), and let \(\lambda = \mbb{E}[\sqrt{\abs{t_2}}]\). Since \(\mbb{E}[\abs{t_2}] = \sqrt{2}\) (see, e.g., \cite{Psarakis:Panaretos:1990}, Theorem~3.1) and since the function \(w \mapsto \sqrt{w}\) on \(\mbb{R}_+\) is strictly concave, it follows by Jensen's inequality that \(\lambda < \sqrt{\mbb{E}[\abs{t_2}]} = 2^{1/4}\). Additionally, recall that \((a + b)^{1/4} \leq a^{1/4} + b^{1/4}\) for all real numbers \(a,b \geq 0\). These facts in conjunction with \eqref{eq: L(W_(n+1)|W_n)} imply that
	\begin{align*}
		\int_\mbb{R} V(w')k(w,w') \, dw' &= \mbb{E}\paren*{\sqrt{\abs{t_2}} \paren*{\frac{1 + (y-w)^2}{2}}^{1/4}} \\
		&= \frac{\lambda}{2^{1/4}} \paren*{1 + (y-w)^2}^{1/4} \\
		&\leq \frac{\lambda}{2^{1/4}} + \frac{\lambda}{2^{1/4}} V(w)
	\end{align*}
	for all \(w \in \mbb{R}\). This establishes the drift condition.
	
	For the minorization condition, fix any \(d > (2^{3/4}\lambda)/(1 - \lambda/2^{1/4})\), and take any \(w \in \mbb{R}\) such that \(V(w) = \sqrt{\abs{y-w}} \leq d\). Then, for all \(w' \in \mbb{R}\),
	\begin{align*}
		k(w,w') &= \frac{1}{\sqrt{8}}\sqrt{\frac{2}{1+(y-w)^2}}\bracket*{1 + \frac{(y-w')^2}{2} \paren*{\frac{2}{{1+(y-w)^2}}}}^{-3/2} \\
		&\geq \frac{1}{\sqrt{8}}\sqrt{\frac{2}{1+d^2}} \bracket*{1 + \frac{(y-w')^2}{2} (2)}^{-3/2} \\
		&= \frac{1}{\sqrt{1+d^2}} \paren*{\frac{1}{\sqrt{8}} \sqrt{2} \bracket*{1 + \frac{(y-w')^2}{2} (2)}^{-3/2}} \\
		&= \frac{1}{\sqrt{1 + d^2}} g(w'),
	\end{align*} 
	where \(g(\cdot)\) is the density given by \eqref{eq: location-scale t_2 density} with \(\mu = y\) and \(\delta = \sqrt{1/2}\). This completes the proof. 
\end{proof}

This example serves an additional purpose: it shows that a full Gibbs sampler may in general fail to inherit geometric ergodicity or a spectral gap from a blocked or collapsed Gibbs sampler. While the blocked Gibbs sampler with operator \(P_{E_U}P_{E_V \intersect E_W}\) has a spectral gap, Theorem~\ref{thm: inheritance of a spectral gap from the full Gibbs sampler} implies that none of the full Gibbs samplers have a spectral gap since there exists a cycle of Gibbs steps without one, namely the blocked Gibbs sampler with operator \(P_{E_V}P_{E_U \intersect E_W}\). Moreover, by Proposition~\ref{prop: U cond. indep. of V given W iff P_(F_U) and P_(F_V) commute}, observe that \(P_{E_U} P_{E_V} P_{E_W} = P_{E_U \intersect E_V} P_{E_W}\) as \(U\) and \(V\) are conditionally independent given \(W\), so the full Gibbs sampler with operator \(P_{E_U} P_{E_V} P_{E_W} = P_{E_U \intersect E_V} P_{E_W}\) can be viewed as a two-component Gibbs sampler that does not have a spectral gap. It follows that the self-adjoint operator of its marginal chain does not have a spectral gap either (Proposition~\ref{prop: spectral connection between two-component deterministic-scan Gibbs and its marginal chain}), and so the marginal chain is not geometrically ergodic (\cite{Roberts:Tweedie:2001}). Hence, the full Gibbs sampler with operator \(P_{E_U} P_{E_V} P_{E_W} = P_{E_U \intersect E_V} P_{E_W}\) is not geometrically ergodic (\cite{Roberts:Rosenthal:2001}, \cite{Diaconis:Khare:Saloff-Coste:2008}), whereas the blocked Gibbs sampler with operator \(P_{E_U}P_{E_V \intersect E_W}\) is.


\begin{appendix}
\section{Proof of Lemma~\ref{lem: P in P_sum has norm less than 1 for some choice of weights implies norm less than 1 for any choice of weights}}

\label{section: proof of lemma for norms related to convex combinations of projections}

Fix any \((w_1', \ldots, w_N') \in \mscr{S}_N\) with \((w_1, \ldots, w_N) \neq (w_1', \ldots, w_N')\), and let \(\gamma = \min_{i \in [N]} w_i'/w_i \in (0,1)\). Further, for \(i \in [N]\), define \(a_i = (w_i' - \gamma w_i)/(1-\gamma)\)  so that \(w_i' = (1-\gamma)a_i + \gamma w_i\). Note that \((a_1, \ldots, a_N) \in \mscr{S}_N\). By the triangle inequality, it follows that
\begin{align*}
	\norm*{\sum_{d = 1}^g w_i' P_{M_i} - P_M} &= \norm*{(1-\gamma)\bracket*{\sum_{d = 1}^g a_iP_{M_i} - P_M} + \gamma \bracket*{\sum_{d = 1}^g w_i P_{M_i} - P_M}} \\
	&\leq (1-\gamma)\norm*{\sum_{d = 1}^g a_iP_{M_i} - P_M} + \gamma \norm*{\sum_{d = 1}^g w_i P_{M_i} - P_M}.
\end{align*}
Since 
\begin{equation*}
	\norm*{\sum_{d = 1}^g a_iP_{M_i} - P_M} = \norm*{\paren*{\sum_{d = 1}^g a_i P_{M_i}} (\Id - P_M)} \leq \norm*{\sum_{d = 1}^g a_i P_{M_i}} \norm{\Id - P_M} \leq 1, 
\end{equation*}
and since \(\norm*{\sum_{d = 1}^g w_iP_{M_i} - P_M} < 1\) by assumption,  
\begin{equation*}
	\norm*{\sum_{d = 1}^g w_i' P_{M_i} - P_M} \leq (1-\gamma)\norm*{\sum_{d = 1}^g a_iP_{M_i} - P_M} + \gamma \norm*{\sum_{d = 1}^g w_i P_{M_i} - P_M} < 1,
\end{equation*}
as required.

\section{Proof of Proposition~\ref{prop: Q(x,dy) is aperiodic if Q is pi-irreducible}}

\label{section: proof of aperiodic proposition}

Suppose \(Q(x,dy)\) were periodic; then, for some integer \(m > 1\), there exists sets \(A_0, \ldots, A_{m-1} \in \mscr{B}(\mcal{X})\) with \(\pi(\bigcup_{j = 0}^{m-1} A_j) = 1\) and \(Q(x, A_{j+1}) = 1\) for all \(x \in A_j\) and all \(j = 0, \ldots, m - 1 \, \pmod m\). Define \(f \in L^2(\pi)\) by \(f(x) = \sum_{j = 0}^{m-1} (e^{2\pi i/m})^{j} \mathbf{1}_{A_j}(x)\) where \(i\) here denotes the imaginary unit, and observe that 
\begin{equation*}
	\norm{f}^2 = \sum_{j = 0}^{m-1} \abs{(e^{2\pi i/m})^j}^2 \pi(A_j) = \sum_{j = 0}^{m - 1} \pi(A_j) = 1,
\end{equation*} 
where the last equality follows from the assumptions that the sets \(A_1, \ldots, A_m\) are mutually disjoint and \(\pi(\bigcup_{j = 0}^{m-1} A_j) = 1\). Furthermore, note that
\begin{align*}
	[Qf](x) &= \sum_{j = 0}^{m-1} (e^{2\pi i/m})^{j} Q(x, A_j) \\
	&= \mathbf{1}_{A_{m-1}}(x) + \sum_{j = 1}^{m-1} (e^{2\pi i/m})^j \mathbf{1}_{A_{j-1}}(x) \\
	&= (e^{2\pi i/m}) f(x)
\end{align*}
for all \(x \in \mcal{X}\). Hence, \(f\) is an eigenfunction for \(Q\) with \(Qf = e^{2\pi i/m} f\).

Consider first the case when \(Q = P_{F_1} \cdots P_{F_g}\) is a cycle of Gibbs steps. Recall that the intersection \(\bigcap_{d = 1}^g F_d = E\) is the closed subspace of constant functions. Furthermore, for every \(d \in [g]\), 
\begin{equation*}
	1 = \norm{f}^2 = \norm{P_{F_d}f}^2 + \norm{(\Id - P_{F_d})f}^2
\end{equation*}
by the Pythagorean theorem. If \(\norm{(\Id - P_{F_d})f} = 0\) for every \(d \in [g]\), then this would mean \(f \in \bigcap_{d = 1}^g F_d = E\), which is nonsensical since \(f\) is not a constant function. Hence, there must exist some \(d' \in [g]\) for which \(\norm{P_{F_{d'}}f} < \norm{f} = 1\), and we may take \(d'\) to be the largest such integer so that \((\Id - P_{F_d})f = 0\) for any \(d > d'\). On the one hand, this implies that
\begin{align*}
	\norm{Qf} &= \norm*{P_{F_1} \cdots P_{F_g}f} = \norm{P_{F_1} \cdots P_{F_{d'}}f} \leq \norm{P_{F_1} \cdots P_{F_{d' - 1}}} \norm{P_{F_{d'}}f} < 1,
\end{align*}
while, on the other hand, \(\norm{Qf} = \norm{e^{2\pi i/m}f} = \abs{e^{2\pi i/m}} \norm{f} = 1\)---a contradiction. The case when \(Q = \sum_{d = 1}^g w_d F_{d}\) is a mixture of Gibbs steps is handled similarly: in this case, 
\begin{equation*}
	\norm{Qf} = \norm*{\sum_{d = 1}^g w_d P_{F_d}f} \leq \sum_{d= 1}^g w_d \norm{P_{F_d} f} \leq w_{d'} \norm{P_{F_{d'}}f} + \sum_{d \neq d'} w_d < 1,
\end{equation*}
which is again a contradiction. Therefore, \(Q(x,dy)\) must be aperiodic.

\end{appendix}

\bibliographystyle{imsart-nameyear} 
\bibliography{refs}       

\begin{thebibliography}{38}

\bibitem[\protect\citeauthoryear{Amit}{1991}]{Amit:1991}
\begin{barticle}[author]
\bauthor{\bsnm{Amit},~\bfnm{Yali}\binits{Y.}}
(\byear{1991}).
\btitle{On rates of convergence of stochastic relaxation for {G}aussian and
  non-{G}aussian distributions}.
\bjournal{J. Multivariate Anal.}
\bvolume{38}
\bpages{82--99}.
\bdoi{10.1016/0047-259X(91)90033-X}
\bmrnumber{1128938}
\end{barticle}
\endbibitem

\bibitem[\protect\citeauthoryear{Badea, Grivaux and
  M\"uller}{2011}]{Badea:Grivaux:Muller:2011}
\begin{barticle}[author]
\bauthor{\bsnm{Badea},~\bfnm{C.}\binits{C.}},
  \bauthor{\bsnm{Grivaux},~\bfnm{S.}\binits{S.}} \AND
  \bauthor{\bsnm{M\"uller},~\bfnm{V.}\binits{V.}}
(\byear{2011}).
\btitle{The rate of convergence in the method of alternating projections}.
\bjournal{Algebra i Analiz}
\bvolume{23}
\bpages{1--30}.
\bdoi{10.1090/S1061-0022-2012-01202-1}
\bmrnumber{2896163}
\end{barticle}
\endbibitem

\bibitem[\protect\citeauthoryear{Bakry, Gentil and
  Ledoux}{2014}]{Bakry:Gentil:Ledoux:2014}
\begin{bbook}[author]
\bauthor{\bsnm{Bakry},~\bfnm{Dominique}\binits{D.}},
  \bauthor{\bsnm{Gentil},~\bfnm{Ivan}\binits{I.}} \AND
  \bauthor{\bsnm{Ledoux},~\bfnm{Michel}\binits{M.}}
(\byear{2014}).
\btitle{Analysis and geometry of {M}arkov diffusion operators}.
\bseries{Grundlehren der mathematischen Wissenschaften [Fundamental Principles
  of Mathematical Sciences]}
\bvolume{348}.
\bpublisher{Springer, Cham}.
\bdoi{10.1007/978-3-319-00227-9}
\bmrnumber{3155209}
\end{bbook}
\endbibitem

\bibitem[\protect\citeauthoryear{Chlebicka, {\L}atuszy{\'{n}}ski and
  Miasojedow}{2025}]{Chlebicka:Latuszynski:Miasojedow:2025}
\begin{barticle}[author]
\bauthor{\bsnm{Chlebicka},~\bfnm{Iwona}\binits{I.}},
  \bauthor{\bsnm{{\L}atuszy{\'{n}}ski},~\bfnm{Krzysztof}\binits{K.}} \AND
  \bauthor{\bsnm{Miasojedow},~\bfnm{B{\l}a{\. {z}}ej}\binits{B.}}
(\byear{2025}).
\btitle{Solidarity of {G}ibbs samplers: the spectral gap}.
\bjournal{Ann. Appl. Probab.}
\bvolume{35}
\bpages{142--157}.
\bdoi{10.1214/24-aap2111}
\bmrnumber{4871706}
\end{barticle}
\endbibitem

\bibitem[\protect\citeauthoryear{Deutsch}{1992}]{Deustch:1992}
\begin{bincollection}[author]
\bauthor{\bsnm{Deutsch},~\bfnm{Frank}\binits{F.}}
(\byear{1992}).
\btitle{The method of alternating orthogonal projections}.
In \bbooktitle{Approximation theory, spline functions and applications
  ({M}aratea, 1991)}.
\bseries{NATO Adv. Sci. Inst. Ser. C: Math. Phys. Sci.}
\bvolume{356}
\bpages{105--121}.
\bpublisher{Kluwer Acad. Publ., Dordrecht}.
\bmrnumber{1165964}
\end{bincollection}
\endbibitem

\bibitem[\protect\citeauthoryear{Deutsch and
  Hundal}{2010}]{Deutsch:Hundal:2010}
\begin{barticle}[author]
\bauthor{\bsnm{Deutsch},~\bfnm{Frank}\binits{F.}} \AND
  \bauthor{\bsnm{Hundal},~\bfnm{Hein}\binits{H.}}
(\byear{2010}).
\btitle{Slow convergence of sequences of linear operators {II}: arbitrarily
  slow convergence}.
\bjournal{J. Approx. Theory}
\bvolume{162}
\bpages{1717--1738}.
\bdoi{10.1016/j.jat.2010.05.002}
\bmrnumber{2718893}
\end{barticle}
\endbibitem

\bibitem[\protect\citeauthoryear{Diaconis, Khare and
  Saloff-Coste}{2008}]{Diaconis:Khare:Saloff-Coste:2008}
\begin{barticle}[author]
\bauthor{\bsnm{Diaconis},~\bfnm{Persi}\binits{P.}},
  \bauthor{\bsnm{Khare},~\bfnm{Kshitij}\binits{K.}} \AND
  \bauthor{\bsnm{Saloff-Coste},~\bfnm{Laurent}\binits{L.}}
(\byear{2008}).
\btitle{Gibbs sampling, exponential families and orthogonal polynomials}.
\bjournal{Statist. Sci.}
\bvolume{23}
\bpages{151--178}.
\bnote{With comments and a rejoinder by the authors}.
\bdoi{10.1214/07-STS252}
\bmrnumber{2446500}
\end{barticle}
\endbibitem

\bibitem[\protect\citeauthoryear{Diaconis, Khare and
  Saloff-Coste}{2010}]{Diaconis:Khare:Saloff-Coste:2010}
\begin{barticle}[author]
\bauthor{\bsnm{Diaconis},~\bfnm{Persi}\binits{P.}},
  \bauthor{\bsnm{Khare},~\bfnm{Kshitij}\binits{K.}} \AND
  \bauthor{\bsnm{Saloff-Coste},~\bfnm{Laurent}\binits{L.}}
(\byear{2010}).
\btitle{Stochastic alternating projections}.
\bjournal{Illinois J. Math.}
\bvolume{54}
\bpages{963--979}.
\bmrnumber{2928343}
\end{barticle}
\endbibitem

\bibitem[\protect\citeauthoryear{Gelfand and Smith}{1990}]{Gelfand:Smith:1990}
\begin{barticle}[author]
\bauthor{\bsnm{Gelfand},~\bfnm{Alan~E.}\binits{A.~E.}} \AND
  \bauthor{\bsnm{Smith},~\bfnm{Adrian F.~M.}\binits{A.~F.~M.}}
(\byear{1990}).
\btitle{Sampling-based approaches to calculating marginal densities}.
\bjournal{J. Amer. Statist. Assoc.}
\bvolume{85}
\bpages{398--409}.
\bmrnumber{1141740}
\end{barticle}
\endbibitem

\bibitem[\protect\citeauthoryear{Geman and Geman}{1984}]{Geman:Geman:1984}
\begin{barticle}[author]
\bauthor{\bsnm{Geman},~\bfnm{Stuart}\binits{S.}} \AND
  \bauthor{\bsnm{Geman},~\bfnm{Donald}\binits{D.}}
(\byear{1984}).
\btitle{Stochastic relaxation, Gibbs distributions, and the Bayesian
  restoration of images}.
\bjournal{IEEE Transactions on Pattern Analysis and Machine Intelligence}
\bvolume{PAMI-6}
\bpages{721-741}.
\bdoi{10.1109/TPAMI.1984.4767596}
\end{barticle}
\endbibitem

\bibitem[\protect\citeauthoryear{Geyer}{2005}]{Geyer:2005}
\begin{bmisc}[author]
\bauthor{\bsnm{Geyer},~\bfnm{Charles~J.}\binits{C.~J.}}
(\byear{2005}).
\btitle{Markov chain Monte Carlo lecture notes}.
\bhowpublished{\url{https://www.stat.umn.edu/geyer/f05/8931/n1998.pdf}}.
\bnote{Lecture Notes, University of Minnesota}.
\end{bmisc}
\endbibitem

\bibitem[\protect\citeauthoryear{Halmos}{1951}]{Halmos:1951}
\begin{bbook}[author]
\bauthor{\bsnm{Halmos},~\bfnm{Paul~R.}\binits{P.~R.}}
(\byear{1951}).
\btitle{Introduction to {H}ilbert space and the theory of spectral
  multiplicity}.
\bpublisher{Chelsea Publishing Co., New York}.
\bmrnumber{45309}
\end{bbook}
\endbibitem

\bibitem[\protect\citeauthoryear{Halperin}{1962}]{Halperin:1962}
\begin{barticle}[author]
\bauthor{\bsnm{Halperin},~\bfnm{Israel}\binits{I.}}
(\byear{1962}).
\btitle{The product of projection operators}.
\bjournal{Acta Sci. Math. (Szeged)}
\bvolume{23}
\bpages{96--99}.
\bmrnumber{141978}
\end{barticle}
\endbibitem

\bibitem[\protect\citeauthoryear{He et~al.}{2016}]{He:De_Sa:Mitliagkas:Re:2016}
\begin{binproceedings}[author]
\bauthor{\bsnm{He},~\bfnm{Bryan~D}\binits{B.~D.}},
  \bauthor{\bsnm{De~Sa},~\bfnm{Christopher~M}\binits{C.~M.}},
  \bauthor{\bsnm{Mitliagkas},~\bfnm{Ioannis}\binits{I.}} \AND
  \bauthor{\bsnm{R\'{e}},~\bfnm{Christopher}\binits{C.}}
(\byear{2016}).
\btitle{Scan order in Gibbs sampling: models in which it matters and bounds on
  how much}.
In \bbooktitle{Advances in Neural Information Processing Systems}
(\beditor{\bfnm{D.}\binits{D.}~\bsnm{Lee}},
  \beditor{\bfnm{M.}\binits{M.}~\bsnm{Sugiyama}},
  \beditor{\bfnm{U.}\binits{U.}~\bsnm{Luxburg}},
  \beditor{\bfnm{I.}\binits{I.}~\bsnm{Guyon}} \AND
  \beditor{\bfnm{R.}\binits{R.}~\bsnm{Garnett}}, eds.)
\bvolume{29}.
\bpublisher{Curran Associates, Inc.}
\end{binproceedings}
\endbibitem

\bibitem[\protect\citeauthoryear{Jones, Roberts and
  Rosenthal}{2014}]{Jones:Roberts:Rosenthal:2014}
\begin{barticle}[author]
\bauthor{\bsnm{Jones},~\bfnm{Galin~L.}\binits{G.~L.}},
  \bauthor{\bsnm{Roberts},~\bfnm{Gareth~O.}\binits{G.~O.}} \AND
  \bauthor{\bsnm{Rosenthal},~\bfnm{Jeffrey~S.}\binits{J.~S.}}
(\byear{2014}).
\btitle{Convergence of conditional {M}etropolis-{H}astings samplers}.
\bjournal{Adv. in Appl. Probab.}
\bvolume{46}
\bpages{422--445}.
\bdoi{10.1239/aap/1401369701}
\bmrnumber{3215540}
\end{barticle}
\endbibitem

\bibitem[\protect\citeauthoryear{Kontoyiannis and
  Meyn}{2012}]{Kontoyiannis:Meyn:2012}
\begin{barticle}[author]
\bauthor{\bsnm{Kontoyiannis},~\bfnm{I.}\binits{I.}} \AND
  \bauthor{\bsnm{Meyn},~\bfnm{S.~P.}\binits{S.~P.}}
(\byear{2012}).
\btitle{Geometric ergodicity and the spectral gap of non-reversible {M}arkov
  chains}.
\bjournal{Probab. Theory Related Fields}
\bvolume{154}
\bpages{327--339}.
\bdoi{10.1007/s00440-011-0373-4}
\bmrnumber{2981426}
\end{barticle}
\endbibitem

\bibitem[\protect\citeauthoryear{Lapidus}{1981}]{Lapidus:1981}
\begin{barticle}[author]
\bauthor{\bsnm{Lapidus},~\bfnm{Michel~L.}\binits{M.~L.}}
(\byear{1981}).
\btitle{Generalization of the {T}rotter-{L}ie formula}.
\bjournal{Integral Equations Operator Theory}
\bvolume{4}
\bpages{366--415}.
\bdoi{10.1007/BF01697972}
\bmrnumber{623544}
\end{barticle}
\endbibitem

\bibitem[\protect\citeauthoryear{Liu}{1994}]{Liu:1994}
\begin{barticle}[author]
\bauthor{\bsnm{Liu},~\bfnm{Jun~S.}\binits{J.~S.}}
(\byear{1994}).
\btitle{The collapsed {G}ibbs sampler in {B}ayesian computations with
  applications to a gene regulation problem}.
\bjournal{J. Amer. Statist. Assoc.}
\bvolume{89}
\bpages{958--966}.
\bmrnumber{1294740}
\end{barticle}
\endbibitem

\bibitem[\protect\citeauthoryear{Liu, Wong and Kong}{1994}]{Liu:Wong:Kong:1994}
\begin{barticle}[author]
\bauthor{\bsnm{Liu},~\bfnm{Jun~S.}\binits{J.~S.}},
  \bauthor{\bsnm{Wong},~\bfnm{Wing~Hung}\binits{W.~H.}} \AND
  \bauthor{\bsnm{Kong},~\bfnm{Augustine}\binits{A.}}
(\byear{1994}).
\btitle{Covariance structure of the {G}ibbs sampler with applications to the
  comparisons of estimators and augmentation schemes}.
\bjournal{Biometrika}
\bvolume{81}
\bpages{27--40}.
\bdoi{10.1093/biomet/81.1.27}
\bmrnumber{1279653}
\end{barticle}
\endbibitem

\bibitem[\protect\citeauthoryear{Meyn and Tweedie}{2009}]{Meyn:Tweedie:2009}
\begin{bbook}[author]
\bauthor{\bsnm{Meyn},~\bfnm{Sean}\binits{S.}} \AND
  \bauthor{\bsnm{Tweedie},~\bfnm{Richard~L.}\binits{R.~L.}}
(\byear{2009}).
\btitle{Markov chains and stochastic stability},
\bedition{2nd} ed.
\bpublisher{Cambridge University Press, Cambridge}.
\bdoi{10.1017/CBO9780511626630}
\bmrnumber{2509253}
\end{bbook}
\endbibitem

\bibitem[\protect\citeauthoryear{Mira and Geyer}{1999}]{Mira:Geyer:1999}
\begin{btechreport}[author]
\bauthor{\bsnm{Mira},~\bfnm{Antonietta}\binits{A.}} \AND
  \bauthor{\bsnm{Geyer},~\bfnm{Charles~J.}\binits{C.~J.}}
(\byear{1999}).
\btitle{Ordering Monte Carlo Markov chains.}
\btype{Technical Report} No. \bnumber{632},
\bpublisher{School of Statistics, University of Minnesota}.
\end{btechreport}
\endbibitem

\bibitem[\protect\citeauthoryear{Murphy}{1990}]{Murphy:1990}
\begin{bbook}[author]
\bauthor{\bsnm{Murphy},~\bfnm{Gerard~J.}\binits{G.~J.}}
(\byear{1990}).
\btitle{{$C^*$}-algebras and operator theory}.
\bpublisher{Academic Press, Inc., Boston, MA}.
\bmrnumber{1074574}
\end{bbook}
\endbibitem

\bibitem[\protect\citeauthoryear{Papaspiliopoulos and
  Roberts}{2008}]{Papaspiliopoulos:Roberts:2008}
\begin{barticle}[author]
\bauthor{\bsnm{Papaspiliopoulos},~\bfnm{Omiros}\binits{O.}} \AND
  \bauthor{\bsnm{Roberts},~\bfnm{Gareth~O.}\binits{G.~O.}}
(\byear{2008}).
\btitle{Stability of the {G}ibbs sampler for {B}ayesian hierarchical models}.
\bjournal{Ann. Statist.}
\bvolume{36}
\bpages{95--117}.
\bdoi{10.1214/009053607000000749}
\bmrnumber{2387965}
\end{barticle}
\endbibitem

\bibitem[\protect\citeauthoryear{Psarakis and
  Panaretos}{1990}]{Psarakis:Panaretos:1990}
\begin{barticle}[author]
\bauthor{\bsnm{Psarakis},~\bfnm{S.}\binits{S.}} \AND
  \bauthor{\bsnm{Panaretos},~\bfnm{J.}\binits{J.}}
(\byear{1990}).
\btitle{The folded {$t$} distribution}.
\bjournal{Comm. Statist. Theory Methods}
\bvolume{19}
\bpages{2717--2734}.
\bdoi{10.1080/03610929008830342}
\bmrnumber{1086030}
\end{barticle}
\endbibitem

\bibitem[\protect\citeauthoryear{Qin, Hobert and
  Khare}{2019}]{Qin:Hobert:Khare:2019}
\begin{barticle}[author]
\bauthor{\bsnm{Qin},~\bfnm{Qian}\binits{Q.}},
  \bauthor{\bsnm{Hobert},~\bfnm{James~P.}\binits{J.~P.}} \AND
  \bauthor{\bsnm{Khare},~\bfnm{Kshitij}\binits{K.}}
(\byear{2019}).
\btitle{Estimating the spectral gap of a trace-class {M}arkov operator}.
\bjournal{Electron. J. Stat.}
\bvolume{13}
\bpages{1790--1822}.
\bdoi{10.1214/19-EJS1563}
\bmrnumber{3959873}
\end{barticle}
\endbibitem

\bibitem[\protect\citeauthoryear{Qin and Jones}{2022}]{Qin:Jones:2022}
\begin{barticle}[author]
\bauthor{\bsnm{Qin},~\bfnm{Qian}\binits{Q.}} \AND
  \bauthor{\bsnm{Jones},~\bfnm{Galin~L.}\binits{G.~L.}}
(\byear{2022}).
\btitle{Convergence rates of two-component {MCMC} samplers}.
\bjournal{Bernoulli}
\bvolume{28}
\bpages{859--885}.
\bdoi{10.3150/21-bej1369}
\bmrnumber{4388922}
\end{barticle}
\endbibitem

\bibitem[\protect\citeauthoryear{Reich}{1983}]{Reich:1983}
\begin{barticle}[author]
\bauthor{\bsnm{Reich},~\bfnm{Simeon}\binits{S.}}
(\byear{1983}).
\btitle{A limit theorem for projections}.
\bjournal{Linear and Multilinear Algebra}
\bvolume{13}
\bpages{281--290}.
\bdoi{10.1080/03081088308817526}
\bmrnumber{700890}
\end{barticle}
\endbibitem

\bibitem[\protect\citeauthoryear{Roberts and
  Rosenthal}{1997}]{Roberts:Rosenthal:1997}
\begin{barticle}[author]
\bauthor{\bsnm{Roberts},~\bfnm{Gareth~O.}\binits{G.~O.}} \AND
  \bauthor{\bsnm{Rosenthal},~\bfnm{Jeffrey~S.}\binits{J.~S.}}
(\byear{1997}).
\btitle{Geometric ergodicity and hybrid {M}arkov chains}.
\bjournal{Electron. Comm. Probab.}
\bvolume{2}
\bpages{13--25}.
\bdoi{10.1214/ECP.v2-981}
\bmrnumber{1448322}
\end{barticle}
\endbibitem

\bibitem[\protect\citeauthoryear{Roberts and
  Rosenthal}{2001}]{Roberts:Rosenthal:2001}
\begin{barticle}[author]
\bauthor{\bsnm{Roberts},~\bfnm{Gareth~O.}\binits{G.~O.}} \AND
  \bauthor{\bsnm{Rosenthal},~\bfnm{Jeffrey~S.}\binits{J.~S.}}
(\byear{2001}).
\btitle{Markov chains and de-initializing processes}.
\bjournal{Scand. J. Statist.}
\bvolume{28}
\bpages{489--504}.
\bdoi{10.1111/1467-9469.00250}
\bmrnumber{1858413}
\end{barticle}
\endbibitem

\bibitem[\protect\citeauthoryear{Roberts and
  Rosenthal}{2015}]{Roberts:Rosenthal:2015}
\begin{barticle}[author]
\bauthor{\bsnm{Roberts},~\bfnm{Gareth~O.}\binits{G.~O.}} \AND
  \bauthor{\bsnm{Rosenthal},~\bfnm{Jeffrey~S.}\binits{J.~S.}}
(\byear{2015}).
\btitle{Surprising convergence properties of some simple Gibbs samplers under
  various scans}.
\bjournal{International Journal of Statistics and Probability}
\bvolume{5}
\bpages{51--60}.
\end{barticle}
\endbibitem

\bibitem[\protect\citeauthoryear{Roberts and Sahu}{1997}]{Roberts:Sahu:1997}
\begin{barticle}[author]
\bauthor{\bsnm{Roberts},~\bfnm{G.~O.}\binits{G.~O.}} \AND
  \bauthor{\bsnm{Sahu},~\bfnm{S.~K.}\binits{S.~K.}}
(\byear{1997}).
\btitle{Updating schemes, correlation structure, blocking and parameterization
  for the {G}ibbs sampler}.
\bjournal{J. Roy. Statist. Soc. Ser. B}
\bvolume{59}
\bpages{291--317}.
\bdoi{10.1111/1467-9868.00070}
\bmrnumber{1440584}
\end{barticle}
\endbibitem

\bibitem[\protect\citeauthoryear{Roberts and
  Tweedie}{2001}]{Roberts:Tweedie:2001}
\begin{barticle}[author]
\bauthor{\bsnm{Roberts},~\bfnm{Gareth~O.}\binits{G.~O.}} \AND
  \bauthor{\bsnm{Tweedie},~\bfnm{Richard~L.}\binits{R.~L.}}
(\byear{2001}).
\btitle{Geometric {$L^2$} and {$L^1$} convergence are equivalent for reversible
  {M}arkov chains}.
\bjournal{J. Appl. Probab.}
\bvolume{38}
\bpages{37--41}.
\bdoi{10.1239/jap/1085496589}
\bmrnumber{1915532}
\end{barticle}
\endbibitem

\bibitem[\protect\citeauthoryear{Rosenthal}{1995}]{Rosenthal:1995}
\begin{barticle}[author]
\bauthor{\bsnm{Rosenthal},~\bfnm{Jeffrey~S.}\binits{J.~S.}}
(\byear{1995}).
\btitle{Minorization conditions and convergence rates for {M}arkov chain
  {M}onte {C}arlo}.
\bjournal{J. Amer. Statist. Assoc.}
\bvolume{90}
\bpages{558--566}.
\bmrnumber{1340509}
\end{barticle}
\endbibitem

\bibitem[\protect\citeauthoryear{Rosenthal and
  Rosenthal}{2015}]{Rosenthal:Rosenthal:2015}
\begin{barticle}[author]
\bauthor{\bsnm{Rosenthal},~\bfnm{Jeffrey~S.}\binits{J.~S.}} \AND
  \bauthor{\bsnm{Rosenthal},~\bfnm{Peter}\binits{P.}}
(\byear{2015}).
\btitle{Spectral bounds for certain two-factor non-reversible {MCMC}
  algorithms}.
\bjournal{Electron. Commun. Probab.}
\bvolume{20}
\bpages{1--10}.
\bdoi{10.1214/ECP.v20-4528}
\bmrnumber{3434208}
\end{barticle}
\endbibitem

\bibitem[\protect\citeauthoryear{Tan and Hobert}{2009}]{Tan:Hobert:2009}
\begin{barticle}[author]
\bauthor{\bsnm{Tan},~\bfnm{Aixin}\binits{A.}} \AND
  \bauthor{\bsnm{Hobert},~\bfnm{James~P.}\binits{J.~P.}}
(\byear{2009}).
\btitle{Block {G}ibbs sampling for {B}ayesian random effects models with
  improper priors: convergence and regeneration}.
\bjournal{J. Comput. Graph. Statist.}
\bvolume{18}
\bpages{861--878}.
\bdoi{10.1198/jcgs.2009.08153}
\bmrnumber{2598033}
\end{barticle}
\endbibitem

\bibitem[\protect\citeauthoryear{Tierney}{1994}]{Tierney:1994}
\begin{barticle}[author]
\bauthor{\bsnm{Tierney},~\bfnm{Luke}\binits{L.}}
(\byear{1994}).
\btitle{Markov chains for exploring posterior distributions}.
\bjournal{Ann. Statist.}
\bvolume{22}
\bpages{1701--1762}.
\bnote{With discussion and a rejoinder by the author}.
\bdoi{10.1214/aos/1176325750}
\bmrnumber{1329166}
\end{barticle}
\endbibitem

\bibitem[\protect\citeauthoryear{van Dyk and Park}{2008}]{van_Dyk:Park:2008}
\begin{barticle}[author]
\bauthor{\bparticle{van} \bsnm{Dyk},~\bfnm{David~A.}\binits{D.~A.}} \AND
  \bauthor{\bsnm{Park},~\bfnm{Taeyoung}\binits{T.}}
(\byear{2008}).
\btitle{Partially collapsed {G}ibbs samplers: theory and methods}.
\bjournal{J. Amer. Statist. Assoc.}
\bvolume{103}
\bpages{790--796}.
\bdoi{10.1198/016214508000000409}
\bmrnumber{2524010}
\end{barticle}
\endbibitem

\bibitem[\protect\citeauthoryear{von Neumann}{1949}]{von_Neumann:1949}
\begin{barticle}[author]
\bauthor{\bparticle{von} \bsnm{Neumann},~\bfnm{John}\binits{J.}}
(\byear{1949}).
\btitle{On rings of operators. {R}eduction theory}.
\bjournal{Ann. of Math. (2)}
\bvolume{50}
\bpages{401--485}.
\bdoi{10.2307/1969463}
\bmrnumber{29101}
\end{barticle}
\endbibitem

\end{thebibliography}


\end{document}